\documentclass{article}

\usepackage{microtype}
\usepackage{graphicx}
\usepackage{subfigure}
\usepackage{booktabs} 

\usepackage{hyperref}

\usepackage[accepted]{icml2024}

\usepackage{amsmath}
\usepackage{amssymb}
\usepackage{mathtools}
\usepackage{amsthm}
\usepackage{stmaryrd}
\usepackage{bm} 

\usepackage[capitalize,noabbrev]{cleveref}

\theoremstyle{plain}
\newtheorem{theorem}{Theorem}[section]

\newtheorem{lemma}[theorem]{Lemma}
\newtheorem{corollary}[theorem]{Corollary}
\theoremstyle{definition}
\newtheorem{definition}[theorem]{Definition}

\theoremstyle{remark}
\newtheorem{remark}[theorem]{Remark}

\newtheorem*{theorem*}{Theorem}
\newtheorem*{lemma*}{Lemma}
\newtheorem{fact}[theorem]{Fact}
\newtheorem*{fact*}{Fact}
\newtheorem*{proposition*}{Proposition}
\newtheorem*{corollary*}{Corollary}

\newtheorem*{hypothesis*}{Hypothesis}

\newtheorem*{conjecture*}{Conjecture}
\theoremstyle{definition}
\newtheorem*{definition*}{Definition}

\newtheorem*{construction*}{Construction}

\newtheorem*{example*}{Example}

\newtheorem*{question*}{Question}
\newtheorem*{assumption*}{Assumption}

\newtheorem*{problem*}{Problem}
\newtheorem{model}[theorem]{Model}
\newtheorem*{model*}{Model}
\theoremstyle{remark}

\newtheorem*{claim*}{Claim}
\newtheorem*{remark*}{Remark}

\newtheorem*{observation*}{Observation}

\usepackage[textsize=tiny]{todonotes}

\usepackage{boxedminipage}

\newcommand{\paren}[1]{(#1)}
\newcommand{\Paren}[1]{\left(#1\right)}

\newcommand{\Brac}[1]{\left[#1\right]}

\newcommand{\Abs}[1]{\left\lvert#1\right\rvert}

\newcommand{\card}[1]{\lvert#1\rvert}
\newcommand{\Card}[1]{\left\lvert#1\right\rvert}

\newcommand{\Set}[1]{\left\{#1\right\}}

\newcommand{\Norm}[1]{\left\lVert#1\right\rVert}

\newcommand{\Snorm}[1]{\Norm{#1}^2}

\newcommand{\iprod}[1]{\langle#1\rangle}

\newcommand{\Esymb}{\mathbb{E}}

\DeclareMathOperator*{\E}{\Esymb}

\newcommand{\suchthat}{\;\middle\vert\;}

\newcommand{\sge}{\succeq}
\newcommand{\sle}{\preceq}

\newcommand{\super}[2]{#1^{\paren{#2}}}

\newcommand\bdot\bullet

\DeclareMathOperator{\Tr}{Tr}

\DeclareMathOperator{\poly}{poly}

\DeclareMathOperator{\polylog}{polylog}

\newcommand{\R}{\mathbb R}

\newcommand{\important}[1]{\texorpdfstring{\textup{\textsc{#1}}}{#1}}
\newcommand{\oracle}{{\important{oracle}} }

\newcommand{\cS}{\mathcal S}
\newcommand{\cT}{\mathcal T}

\newcommand{\cX}{\mathcal X}

\newcommand{\bbP}{\mathbb P}

\renewcommand{\leq}{\leqslant}

\renewcommand{\geq}{\geqslant}
\renewcommand{\ge}{\geqslant}
\let\epsilon=\varepsilon

\newcommand{\eps}{\epsilon}

\allowdisplaybreaks
\newcommand*{\Id}{\mathrm{Id}}

\DeclareMathOperator{\diag}{diag}

\newcommand*{\transpose}[1]{{#1}{}^{\mkern-1.5mu\mathsf{T}}}

\renewcommand{\ij}{{ij}}

\icmltitlerunning{Fast Algorithm for Beyond-Worst-Case Graph Clustering}

\begin{document}
\onecolumn
\icmltitle{A Near-Linear Time Approximation Algorithm for Beyond-Worst-Case Graph Clustering}



\icmlsetsymbol{equal}{*}

\begin{icmlauthorlist}
\icmlauthor{Vincent Cohen-Addad}{equal,google}
\icmlauthor{Tommaso d'Orsi}{equal,google,bocconi}
\icmlauthor{Aida Mousavifar}{equal,google}
\end{icmlauthorlist}

\icmlaffiliation{google}{Google Research}
\icmlaffiliation{bocconi}{BIDSA, Bocconi}

\icmlkeywords{Balanced cut, near-linear time, graph clustering, semi-random models, Machine Learning, matrix multiplicative weights, ICML}

\vskip 0.3in



\printAffiliationsAndNotice{\icmlEqualContribution} 

\begin{abstract}
We consider the semi-random graph model of \cite{DBLP:conf-stoc-MakarychevMV12}, where, given a random bipartite graph with $\alpha$ edges and an unknown bipartition $(A, B)$ of the vertex set, an adversary can add arbitrary edges inside each community and remove arbitrary edges from the cut $(A, B)$ (i.e. all adversarial changes are \textit{monotone} with respect to the bipartition). 
For this model, a polynomial time algorithm is known to approximate the Balanced Cut problem up to value $O(\alpha)$ \cite{DBLP:conf-stoc-MakarychevMV12} as long as the cut $(A, B)$ has size $\Omega(\alpha)$. However, it consists of slow subroutines requiring optimal solutions for logarithmically many semidefinite programs.
We study the fine-grained complexity of the problem and present the first near-linear time algorithm that
achieves similar performances to that of \cite{DBLP:conf-stoc-MakarychevMV12}. Our algorithm runs in time $O(|V(G)|^{1+o(1)} + |E(G)|^{1+o(1)})$ and finds a balanced cut of value $O(\alpha)\,.$
Our approach appears easily extendible to related problem, such as Sparsest Cut,  and also yields an near-linear time $O(1)$-approximation to Dagupta's objective
function for hierarchical clustering \cite{Dasgupta16} for the semi-random hierarchical stochastic
block  model inputs of \cite{cohen2019hierarchical}.
\end{abstract}


\section{Introduction}\label{section:introduction}
Graph clustering and partitioning problems are central in combinatorial optimization. Their study has led to a large variety of
key results, leading to new fundamental ideas and impactful practical outcomes.
The sparsest cut and balanced cut problems are iconic examples: On the one hand, they have served as a testbed for designing new breakthrough
algorithmic techniques, from the seminal paper of Leighton and Rao~\cite{DBLP:conf-focs-/LeightonR88} up to the results of Arora, Rao, and Vazirani~\cite{DBLP:conf-stoc-AroraRV04}
and Sherman~\cite{DBLP:conf-focs-Sherman09}. On the other hand, they are models for graph partitioning problems in various data mining
and unsupervised machine learning applications and have thus inspired widely-used heuristics in more applied fields.

\paragraph*{Beyond worst-case instances}
A frustrating gap exists between the impressive theoretical results obtained over the last three decades and the
success of heuristics used in practice. While poly-logarithmic approximation algorithms have been developed for balanced cut
and sparsest cut (and related problem such as minimum bisection~\cite{racke2008optimal}, multicut~\cite{garg1996approximate}, min uncut~\cite{goemans1995improved, agarwal2005log}), the algorithm design
community has had little success in obtaining constant factor approximation algorithms for these problems. In fact, the Unique
Games Conjecture even suggests that such bounds may be very hard to obtain~\cite{khot2007optimal, khot2015unique, raghavendra2008optimal, raghavendra2012reductions}. Thus, to be able to show good approximation bounds
and design algorithms that are tailored to real-world instances, one must shift the focus from the worst-case to the
so called \emph{beyond-worst-case} complexity of the problems.
  
This conclusion has seeded a long line of work  aimed at modeling
average instances encountered in practice and designing algorithms for these models~\cite{dyer1986fast, bui1987graph, boppana1987eigenvalues, feige2001heuristics, mcsherry2001spectral} (or analyzing existing algorithms in
these models~\cite{jerrum1993simulated, dimitriou1998go, bilu2012stable}). 
For the model to be relevant it should forbid pathological instances that are
extremely unlikely in practice while capturing the essence of the real-world instances without oversimplifying them.

While there has been a significant amount of work on inference in random and semi-random graph models,
the work of Makarychev, Makarychev and Vijayaraghavan~\cite{DBLP:conf-stoc-MakarychevMV12} is among the first to analyze the
approximability and complexity of the graph partitioning objectives mentioned above for extremely general families of semi-random graphs.
In their settings, the input is generated from a distribution over graphs that exhibit a cluster structure. Concretely, the graph consists of
two communities and a planted random cut between the communities,  the adversary can modify the graph in agreement with the cluster structure
by arbitrarily adding edges within the communites and / or sparsifying the random cut across communities,\footnote{These are often times referred to as \textit{monotone} perturbations. Such perturbations may have surprising effects on the statistical and computational aspects of the problem. For instance see \cite{moitra2016robust, liu2022minimax}.} see \cref{model:main} for a precise definition.
In this context, the goal is not to recover the underlying cluster structure -- which may be information-theoretically impossible --
but rather to provide a good approximation to the cut objectives.

The motivation for studying such models is the following. In practice, the graphs we aim at clustering have an
unknown underlying cluster structure that we would like to identify -- and that's why we are running a clustering algorithm in
the first place. In this context, on the one hand the intra-cluster topology may be very peculiar and so possibly adversarial (hence we
would like to let the adversary freely choose
the intra-cluster topology\footnote{We remark  this model is significantly more general than the stochastic block model, see \cref{section:related-research}}), on the other hand the inter-cluster topology is often more random, sometimes interpreted as noise between clusters and hence modeled as a random cut,
see also the discussion and motivating examples provided in~\cite{DBLP:conf-stoc-MakarychevMV12}.

Of course, allowing the adversary to make the 
planted cut denser -- and by doing so to smooth out the underlying cluster structure --
would bring us back to the worst-case setting; the semi-random model proposed above is thus a step in between.


Hence, with the idea of bridging the gap between worst-case complexity of the problems and real-world instances,
Makarychev, Makarychev, Vijayaraghavan~\cite{DBLP:conf-stoc-MakarychevMV12} developed a general algorithmic framework for 
graph partitioning problems in the above semi-random instances which
achieves an $O(1)$-approximation algorithm (for a wide array of parameters) for balanced cut and sparsest cut and related problems
such as multicut, min uncut and small set expansion.

While the result of~\cite{DBLP:conf-stoc-MakarychevMV12} is close to optimal in the sense that it achieves an $O(1)$-approximation
for several classic graph partitioning problems and a wide range of parameters, it relies on an heavy machinery that requires to iteratively solve multiple semi-definite programs.
In fact the running time is not stated in the paper and seem to require $\Omega(n^3)$ time for the rounding on top of the time it takes to obtain optimal solutions to
polylogarithmically many semi-definite programs with more than $\Omega(n^3)$ constraints.\footnote{We point out that the algorithm  \textit{requires} an \textit{actual} feasible solution with nearly optimal objective value and not a rounded solution.}
We initiate the study of the \emph{fine-grained} complexity of the problem and ask: \textit{How fast can we solve beyond-worst-case
instances (involving semi-random perturbations)?}


\subsection{Results}
Before providing our main theorem, we introduce the model of interest.
\begin{model}[Random cut with monotone perturbations]\label{model:main}
	We consider graphs over $n$ vertices generated through the following process. Let $a \in (0,1/2)$, $\eta(n)\in(0,1)$:
	\begin{enumerate}
		\item[(i)] The adversary partitions $[n]$ into sets $A, B$ satisfying $\card{A}, \card{B}\geq an$.
		\item[(ii)] Each edge between $A$ and $B$ is drawn randomly and independently with probability $\eta$.
		\item[(iii)] The adversary arbitrarily adds edges within $A$ and within $B$.
		\item[(iv)] The adversary arbitrarily removes edges between $A$ and $B$.
	\end{enumerate}
\end{model}

Our main result is an algorithm that, given an instance of \cref{model:main} with a $\Omega(n^2\cdot \eta)$-sized $(A,B)$ cut, returns a $O(1)$-approximation in almost linear time.\footnote{We write $o(1)$ to denote  real-valued functions tending to zero as $n$ grows.}
\begin{theorem}\label{theorem:main}
  Let $G$ be a graph over $n$ vertices generated through \cref{model:main} with parameters $a>0,\eta\geq \Omega(\frac{(\log  n)^2 \cdot (\log\log n)^2}{n})\,.$ 
  There exists an algorithm that on input $G$, with  probability $1-o(1)$, outputs an $\Omega(a)$-balanced cut of value at most $O(n^2\cdot \eta)$, namely a cut where each side has size at least $\Omega(a \cdot n)$.
  
  Moreover, the algorithm runs in time $O\Paren{\Card{V(G)}^{1+o(1)}+\Card{E(G)}^{1+o(1)}}$. 
\end{theorem}

\cref{theorem:main} is a significant step toward bridging the gap between the theoretically-oriented work of~\cite{DBLP:conf-stoc-MakarychevMV12} and the practical motivation behind semi-random models.
The error guarantees of the underlying algorithm match those of \cite{DBLP:conf-stoc-MakarychevMV12}, but the running time is nearly linear.
Despite the fact that further steps remains to be taken to provide algorithmic solutions that both matches the theoretical guarantees of~\cite{DBLP:conf-stoc-MakarychevMV12} and whose running time is competitive with state-of-the-art Bisection heuristics, our algorithm is a \textit{first} example that general beyond-worst-case graph clustering can be done in near linear time.

Finally, we believe that understanding the fine-grained complexity of balanced cut and related problems beyond-the-worst case is an important line of work and the techniques presented here could lead to further improvements for other related problems for which the beyond-worst-case analysis has been studied (e.g.: Bilu-Linial stability for multicut~\cite{BiluL12,AngelidakisMM17}).

\paragraph*{Generalizations}
Our approach appears to also be easily extendable to other graph problems. As a concrete example, we consider the \textit{semi-random hierarchical stochastic block model} (henceforth HSM) of~\cite{cohen2019hierarchical}. In~\cite{cohen2019hierarchical}, the authors studied the celebrated objective
function for hierarchical clustering introduced by Dasgupta~\cite{Dasgupta16} and investigate how well it can be approximated
beyond-the-worst-case. Assuming the Small Set Expansion hypothesis \cite{raghavendra2010graph}, the problem cannot be approximated within any constant factor. The authors thus introduce a generative model for hierarchical clustering inputs called the \emph{hierarchical stochastic block model} that naturally generalizes the classic stochastic block model, and show that one can approximate Dasgupta's objective
up to a constant factor in that model and under semi-random perturbation (the precise definition of the model can be found in
\cref{section:semi-random-hsm}). In this paper, we significantly improve the complexity of the algorithm of~\cite{cohen2019hierarchical}.

 \begin{theorem}
\label{theo61}
Let $G$ be a graph generated from the HSM (Definition \ref{Definition51}) with $p_{min}= \Omega\left(\log n / n^{2/3} \right)$. Then, there exists a randomized algorithm that runs in time $O\Paren{\Card{V(G)}^{1+o(1)}+\Card{E(G)}^{1+o(1)}}$ with probability $1-o(1)$ outputs a tree $T$ such that
\begin{equation}
\label{eq:eq14}
    cost(T;G) = O(OPT(\bar{G})),
\end{equation}
where $OPT(\bar{G})$ denotes the value of the optimal tree for $\bar{G}$ and we note that $OPT(\bar{G})=cost(\widetilde{T};\bar{G})$, where $\widetilde{T}$ is the generating tree. Furthermore, the above holds even in the semi-random case, i.e., when an adversary is allowed to remove any subset of the edges from $G$.
\end{theorem}


\subsection{Related Research}\label{section:related-research}
There has been extensive research on graph partitioning problems for random and semi-random models.
Perhaps the most extensively studied example is the stochastic block model (see \cite{abbe2017community} for a broad overview). In its simplest form, the model describes graphs where both the inter-community and the intra-community topologies are random. That is  the graph is randomly partitioned into two subsets $(A, B)$ of the same size such that every edge between the set $A$ and set $B$ exists with probability $\eta$, and edges inside communities $A$, and $B$ exists with probability $\mu\geq \eta$.\footnote{We remark that from both a computational and a statistical point of view, sharp phase transitions appear depending on the relation between the expected average degree  and the community bias. We omit a detailed discussion and refer the interested reader to the aforementioned survey.} 
Many algorithms are known to succesfully recover the partition for typical instances of the model \cite{decelle2011asymptotic,  massoulie2014community, mossel2015reconstruction, hopkins2017efficient, mossel2018proof}. 
In recent years, an ongoing line of work has aimed to extend these algorithmic techniques to more general semi-random models \cite{feige2001heuristics, moitra2016robust, montanari2016semidefinite,  ding2022robust, hua2023reaching}, first by introducing monotone perturbations \cite{feige2001heuristics, moitra2016robust, FeiC19, liu2022minimax} (a perturbation is monotone with respect to the bipartition $(A, B)$ if it adds edges inside the comunities or remove edges accross communities) and then by allowing a small but constant fraction of adversarially chosen edge
 \cite{ding2022robust} or vertex \cite{liu2022minimax, hua2023reaching}  perturbations. These results still crucially rely on the randomness of the intra-cluster topology and thus cannot work in the  significantly more general context of \cref{model:main}, where the structure inside $A$ and $B$ is arbitrarily and not  random. (We remark that random settings with $0<\mu\leq \eta$ have been investigated in the statistical physics literature.\footnote{The model is usually referred to as the \textit{antiferromagnetic model}.} These results also cannot be used in the presence of monotone perturbations.)
 
\cref{model:main} was extended in \cite{MakarychevMV14}, where the same set of authors introduced the so-called PIE model. Here the main assumption concerning the edges across the $(A, B)$ partition is that they were sampled from a permutationally invariant distribution (w.r.t to edges). Their error guarantees are comparable to those of \cite{DBLP:conf-stoc-MakarychevMV12} in the denser regime $\eta\geq O(\tfrac{\polylog(n)}{n})\,.$
We leave it as an open question to extend \cref{theorem:main} to this model.

Less general semi-random models, in which adversarial perturbations are applied before sampling the random edges, have also been studied. Interestingly, for these significantly weaker adversaries, spectral algorithms have been shown to achieve nearly optimal guarantees \cite{ChierichettiPRT22}. However, as already mentioned, these algorithms are known to be fragile to perturbations such as in \cref{model:main} and thus cannot be expected to lead to error guarantees  comparable to those in \cref{theorem:main}.

Further recent work \cite{peng2020robust} developed a sublinear robust algorithm for local reconstruction of noisy $(k, \phi_{\text{in}}, \phi_{\text{out}})$-clusterable graphs. Such graphs consists of $k$ expanders with inner conductance at least $\phi_{\text{in}}$ and outer conductance at most $\phi_{\text{out}}$ and an adversary is allowed to modify at most  $\epsilon$ fraction of edges within clusters. This result holds under the assumption $\phi_{\text{out}} \leq \frac{ \epsilon \cdot \phi_{\text{in}}}{\log n\cdot k^{O(1)}}$.


\section{Techniques}\label{section:techniques}

We present here the main ideas contained in the proof of \cref{theorem:main}.
Throughout the section let ${G}$ be a random graph sampled through steps (i) and (ii) of \cref{model:main} and let $G^\circ$ be the resulting graph after steps (iii) and (iv).  We let $\alpha\leq (1+o(1))\cdot n^2\cdot \eta$ be the number of edges in  ${G}$ and thus an upper bound on the optimal balanced cut in $G^\circ$ with high probability.

\paragraph*{Slow algorithms for balanced cut in the semi-random model} In order to present our techniques for obtaining a near linear time algorithm
with constant approximation factor for \cref{model:main}, it is necessary to first understand how the known \textit{slow} algorithm of \cite{DBLP:conf-stoc-MakarychevMV12} works.
Consider the random graph $G$ and let $v_1,\ldots, v_n\in \R^n$ be the embedding given by the returned
SDP solution, where $v_i$ corresponds to the embedding of the $i$-th vertex of ${G}$.\footnote{See \cref{section:mmw-algorithm} for a definition of the program.}
The algorithm of \cite{DBLP:conf-stoc-MakarychevMV12} is an iterative procedure that cycles over two subroutines: first, the algorithm solves the canonical balanced cut
SDP as in \cite{DBLP:conf-stoc-AroraRV04}; second, the algorithm \textit{carefully} removes clusters of vertices that are particularly
\textit{close} to each other in the embedding given by the found optimal SDP solution.

Concretely, the latter step identifies  so-called $(\delta, n)$-\emph{heavy} vertices, namely vertices $i$ such that its embedding $v_i$
in the SDP solution is at distance at most $\delta$ from at least $10\delta^2 n$ embeddings of other vertices in the SDP
solution\footnote{Notice that a ball of radius $2\delta$ centered at any heavy vertex $v$ contains at least $\delta^2n$ vertices.}.
Then, the algorithm \emph{carves out} a ball of radius $\delta$: It creates a cluster containing $i$ and all the vertices $j$ s.t. $v_j$ is
at distance at most $\delta$ from $v_i$. Here $\delta$ is a
parameter chosen appropriately. At the end of this process, the algorithm has removed a set  $H_\delta\subseteq V({G})$ of vertices from the
instance.

The crucial observation here is that, in \textit{every} feasible embedding of the random graph ${G}$, if the random cut is dense enough, then when \textit{restricted} to
the non-heavy vertices, it will satisfy a one-sided Chebyshev-like inequality of the form:
\begin{align*}
	\bbP_{\ij \overset{u.a.r.}{\sim}E({G}\setminus H_\delta)}\Paren{\Snorm{v_i-v_j}\leq \delta}\leq 1/\delta^2\,.
\end{align*}
That is, with high probability only a $O(\delta^2)$-fraction of the edges between non-heavy vertices is shorter than $\delta$ in the embedding.
This property is called \textit{geometric expansion}.

Now the crux of the argument is that, given a feasible embedding $v_1,\ldots, v_n \in \R^n$ of ${G}$, if by removing heavy vertices we don't cut more than $O(\alpha)$ edges, then the geometric expansion property guarantees that the minimum balanced cut in the \textit{remaining} graph has cardinality at most $O(\delta^2\cdot \alpha)$. Thus after several iterations of the algorithm we have decreased the value of the minimum balanced cut by at least a $O(1/\sqrt{\log n})$ factor (and in fact, we will need to decrease it by a $1/\log n$ factor in order to use the algorithm of \cite{DBLP:conf-focs-Sherman09} in near linear time) and so a simple application of the SDP rounding of \cite{DBLP:conf-stoc-AroraRV04} returns now a cut of optimal value $\alpha$.
Importantly, the geometric expansion property is robust to monotone changes -- namely to the changes that the adversary can make to the graph
at the last two steps of the generative model (\cref{model:main}) and thus, the exact same reasoning applies for the graph  $G^\circ$ as well.

\paragraph*{Roadblocks to speeding up the algorithms via the matrix multiplicative weights framework}
While semidefinite programs are computationally expensive to solve, there is by now a rich literature on fast algorithms to approximately solve them (e.g. see \cite{DBLP:conf-stoc-AroraK07, DBLP:conf-focs-Sherman09, DBLP:conf-soda-Steurer10}, see \cref{section:background} for a comprehensive description).
These results rely on the matrix multiplicative weight method (henceforth MMW) \cite{DBLP:conf-stoc-AroraK07}.\footnote{This can be seen as an application of the mirror descent algorithm with the von Neumann negative entropy as the chosen mirror map.} 
The framework aims at obtaining a "feasible enough" solution to the SDP so that the desired rounding argument works out. 
On a high level, the approach  is based on the following steps: \textit{(i)} find an assignment of the program variables that is only approximately feasible, in the sense that only a subset of constraints is approximately satisfied, \textit{(ii)} round this infeasible solution into a feasible, integral solution.
The key underlying idea is that, for many problems, if the subset of constraints that gets satisfied is chosen carefully, then the rounding algorithm works even though the starting assignment is far from being  feasible.
The running time improvement is obtained by designing  an \oracle algorithm that efficiently answers \textbf{yes}, if the candidate solution is feasible enough, or otherwise answers \textbf{no} and exhibits constraints that are violated by the current solution. 
These will then be used to pick the direction of movement in the underlying mirror descend algorithm.
The running time of the oracle depends on the so-called oracle's \textit{width} (see \cref{section:background}) and thus the challenge is usually to design oracles of bounded width.

In the context of balanced cut for arbitrary graphs, this approach has been extremely successful, leading to an $O(\sqrt{\log n}/\epsilon)$ approximation algorithm running in time $O(n^{1+\eps}\log n+m)$ when combining \cite{DBLP:conf-focs-Sherman09} and maximum flow algorithm of \cite{DBLP:journals-corr-abs-2203-00671}.

With respect to our settings, a natural question to ask is whether  this matrix multiplicative framework may be used at each iteration of the slow algorithm above to improve its running time. While this intuition is --in principle-- correct, to obtain significant running time speed ups,  additional fundamental challenges need to be solved. 
First, the heavy vertices removal procedure, that requires to find out all the heavy vertices of the graph -- or in other words, all the particularly dense balls in the SDP solution-- is slow.   Second and most important, this subroutine crucially relies on having access to an \textit{optimal} SDP solution. Hence the procedure cannot work with the infeasible solutions computed by the MMW framework. 

To overcome these obstacles we need to deviate from the canonical matrix multiplicative weights paradigm.

\paragraph*{Approximate heavy vertices removal}
Our first improvement, thus, consists of designing a faster algorithm to replace the subroutine  that  identifies all the dense balls of the optimal SDP solution. Identifying dense balls in high-dimensional
Euclidean spaces (say of dimension $\Omega(\eps^{-2} \log n)$) is a well studied problem. We make use of subsampling techniques
to \emph{approximately} solve this problem. Namely, our procedure will recover all the heavy vertices but may yield false positives, namely
balls that are almost dense -- up to a constant factor away from the target density. Concretely, we ask for balls that contain at least $10\delta^2 n$
vertices but are of radius $\sqrt{2}\delta$ instead of $\delta$.
While we can achieve this in time $\tilde{O}(\Card{V({G})})$, we now have to modify the next steps of the rounding to take the false positive
into account in the rounding.

\paragraph*{Main challenge: rounding and the matrix multiplicative weights framework}
The second challenge is more significant. Our idea is the introduction of a probabilistic oracle of small width that leverages the geometric expansion of the planted cut.
Concretely, recall in the previous paragraph we introduced an approximate heavy vertex removal procedure with the property that, with constant probability, only a few edges will be cut \emph{if} applied to a feasible solution of small objective value. 
It is important to notice that if the procedure
cuts few edges even when applied to an infeasible solution  of high objective value,  we would still be satisfied
since we would be making good progress in the graph partitioning at low cost. 
That is, even if a candidate solution is overall far from being feasible, it turns out to be sufficiently good for us if it is close to being feasible on the heavy vertices and their neighborhoods.
We thus mainly have to deal with the case where too many edges get cut. The crux of the argument then is that if the procedure cuts too many edges,  we can show how
with reasonable probability there exists a hyperplane of \emph{small width} separating our solution from the set of feasible solutions of small objective value, and moreover we can 
identify it efficiently. This means we can make progress and obtain a better solution through applying a step of the  matrix multiplicative weights framework.

The intuition behind the \textbf{no}-case of this probabilistic oracle is that, \textit{on average}, the probability that an edge with exactly one endpoint in these heavy balls is cut throughout the removal procedure must be larger than for feasible embeddings with small objective value. Indeed otherwise we could have expected our procedure to cut fewer edges.
Thus we can conclude that with constant probability, either the current solution has significantly larger objective value, or several triangle inequalities must be violated at the same time and thus we can provide a feedback matrix of small width.  

To ensure that our heavy vertices removal procedure would have indeed cut fewer edges if given a feasible solution, we repeat this process poly-logarithmically many times.

\begin{remark}[On the minimum edge density $\eta$ in the cut]
	As already briefly discussed, \cref{theorem:main} requires $\eta\geq \Omega(\frac{(\log  n)^2 \cdot (\log\log n)^2}{n})$. In comparison \cite{DBLP:conf-stoc-MakarychevMV12} only requires a lower bound of $\Omega(\frac{(\sqrt{\log n} \cdot (\log\log n)^2}{n})$. This discrepancy comes from the fact that a near linear time budget only allows us to only obtain a $O(\log n)$ approximation to balanced cut using \cite{DBLP:conf-focs-Sherman09, DBLP:journals-corr-abs-2203-00671}.
	We offset these worse guarantees leveraging stronger geometric expansion properties, which may not hold for $\eta < O(\frac{(\log  n)^2 \cdot (\log\log n)^2}{n})\,.$ What density is necessary for constant approximation to be possible, remains a fascinating open question. 
\end{remark}

\subsection{Perspective}
There is also a second perspective from which we may see \cref{theorem:main}.
The last decade has seen tremendous advancements in the design of algorithms for inference problems that are robust to adversarial corruptions (among many we cite  \cite{d2020sparse, ding2022robust,liu2022minimax, bakshi2022robustly, guruswami2022algorithms, hua2023reaching}, see also the survey \cite{diakonikolas2019recent}).
The emerging picture, which unfortunately appears hard to formalize, is that certain algorithmic techniques --such as semidefinite programming-- appear robust while others --such as low-degree polynomials or local search-- can be fooled by carefully chosen perturbations. 
In particular, the use of optimal solutions for semidefinite programs have been a fundamental tool behind these results.
Unfortunately, the computational budget required to find such objects is often large, making comparable results  hard to achieve in near linear time.
\cref{theorem:main} is a first significant example in which one can retain this \textit{"robustness property"} while working with a near-linear time computational budget, hence providing a first step towards more practical robust algorithms.

\section{Organization and notation}
The rest of the paper is organized as follows.
In \cref{section:mmw-algorithm} we present the algorithm. We provide necessary background on the matrix multiplicative framework and on other notions used throughout the paper in \cref{section:background}.
We present our \oracle in \cref{section:oracle}, which combined with the results in \cref{section:mmw-algorithm} yields  \cref{theorem:main}.
We then study the hierarchical stochastic block  in \cref{section:semi-random-hsm}.

We hide multiplicative factors \textit{poly-logarithmic} in $n$ using the notation $\tilde{O}(\cdot)\,, \tilde{\Omega}(\cdot)$. Similarly, we hide absolute constant multiplicative factors using the standard notation $O(\cdot)\,, \Omega(\cdot)\,, \Theta(\cdot)$. Often times we use the letter $C$ to denote universal constants independent of the parameters at play. 
Given a function $g:\R\rightarrow \R$, we write $o(g)$ for the set of real-valued functions $f$ such that $\lim_{n\rightarrow \infty}\frac{f(n)}{g(n)}=0$. Similarly, we write $g\in \omega(f)$ if $f\in o(g)$.
Throughout the paper, when we say "an algorithm runs in time $O(q)$" we mean that the number of basic arithmetic operations involved is $O(q)$. That is, we ignore bit complexity issues.

\paragraph*{Vectors and matrices}We use $\Id_n$ to denote the $n$-by-$n$ dimensional matrix and $\mathbf{0}$ to denote the zero matrix. For matrices $A, B \in \R^{n\times n}$ we write $A\sge B$ if $A-B$ is positive semidefinite. For a matrix $M$, we denote its eigenvalues by  $\lambda_1(M)\,,\ldots, \lambda_n(M)$; we simply write $\lambda_i$ when the context is clear. We denote by $\Norm{M}$ the spectral norm of $M$.
Let $\cS_n\subset\R^n$ be the set of real symmetric $n$-by-$n$ matrices and let $\Delta_n(r):=\Set{X\in \cS_n\suchthat \Tr X\leq r\,, X\sge 0}$. For $X\in \cS_n$, the matrix exponential is $\exp(X)=\sum_{i=0}^{\infty}\frac{X^i}{i!}$. We remark that $\exp(X)$ is positive semidefinite for all symmetric $X$ as $\exp(X)=\transpose{\Paren{\exp(\tfrac{1}{2}X)}}\exp(\tfrac{1}{2}X)$.
For a vector $v\in \R^n$, we write $v\geq \mathbf{0}$ if all entries of $v$ are non-negative. We use $\mathbb{S}^n\subseteq \R^n$ to denote the unit sphere.

\paragraph*{Graphs} We denote graphs with the notation $G(V, E)$. We use $V(G)$ to denote the set of vertices in $G$ and similarly $E(G)$ to denote its set of edges. For a graph $G$  we write $L_G$ for the  associated combinatorial Laplacian, which is a matrix with rows and columns indexed by the nodes of $G$ such that $(L_G)_{ii} = \sum_{\ij\in E(G)} 1$, i.e. the  degree of node $i$, and for $i\neq j\,$ $(L_G)_\ij$ is $-1$ if $\ij \in E(G)$ and $0$ otherwise. When the context is clear we drop the specification of $G$. Unless specified otherwise, we use $n$ to denote $\Card{V(G)}$. For a partition  $(A, B)$ of the vertices of $G$, we write $E(A, B)\subseteq E$ for the set of edges in the $A$-$B$ cut. We say that a partition $(A,B)$ is $a$-balanced if $\Card{A}/\Card{B}\geq a$ assuming $\Card{A}\leq \Card{B}$.

\section{A fast algorithm for semi-random balanced cut}\label{section:mmw-algorithm}

We present here our main theorem which implies  \cref{theorem:main}.
To solve the balanced cut problem we  consider its basic SDP relaxation. 
Given a graph $G$, the  relaxation for the $a$-balanced cut problem is:
\begin{equation}\label{eq:balanced-cut-sdp}
	\Set{
		\begin{aligned}
			\min &\sum_{ij\in E} c_\ij \Snorm{v_i-v_j}\\
			&\Snorm{v_i}=1 &\forall i \in [n] & \\
			&\Snorm{v_i-v_j}+\Snorm{v_j-v_k}&&\\
			&\qquad\qquad\quad\geq \Snorm{v_i-v_k}&\forall i,j,k,\in[n] &\\
			&\sum_{i,j \in [n]} \Snorm{v_i-v_j}\geq 4an^2&&
		\end{aligned}
	}
\end{equation} 
We refer to the first constraint as the unit norm constraint, to the second as the triangle inequality constraint and to the third as the balance constraint.
\cref{eq:balanced-cut-sdp} may be rewritten in its canonical form
\begin{equation}\label{eq:canonical-balanced-cut-primal-sdp}
	\Set{
		\begin{aligned}
			\min &\iprod{L,X}\\
			&X_{ii}=1&\forall i \in [n] & \\
			&\iprod{T_p, X}\geq 0&\forall \text{ paths $p$ of length 2} &\\
			&\iprod{K_V, X}\geq 4an^2&&
		\end{aligned}
	}
\end{equation} 
where $L$ is the combinatorial Laplacian of the graph, $K_S$ is the Laplacian of the complete graph over vertex set $S$ and, for a path $p$, $T_p$ is the difference between the Laplacian of $p$ and the Laplacian of the single edge connecting its endpoints. 
Notice that $v_1,\ldots,v_n$ are the Gram vectors  of $X$. To ease the reading we will sometimes use the vectors representation and others the matrix representation. We denote by $\alpha$ the optimal value for a given instance of \cref{eq:canonical-balanced-cut-primal-sdp}. In particular, we say a graph $G$ has optimal cut $\alpha$ if minimum solutions to \cref{eq:balanced-cut-sdp} have objective value $\alpha$.  Notice that for graphs generated as in \cref{model:main}, with high probability we have $\alpha \leq (1+ o(1))n^2\cdot \eta\,.$

Before stating the main theorem we require a couple of definitions concerning the embedding of graphs. These are based on  \cite{DBLP:conf-stoc-MakarychevMV12}.

\begin{definition}[Heavy vertex]\label{definition:heavy-vertex}
	Let $\delta,  n,n'>0$.
	Let $G(V, E)$ be a graph on $n$ vertices and let $X$ be the Gram matrix of an embedding $v_1,\ldots,v_n\in\R^n$  of $G$ onto $\R^n$. A vertex $i\in V$ is said to be $(\delta, n')$-heavy if
	\begin{align*}
		\Card{\Set{j \in V(G)\suchthat \Snorm{v_i-v_j}\leq \delta}} \geq \delta^2\cdot n'\,.
	\end{align*}
	We denote the set of $(\delta, n')$-heavy vertices in the embedding $X$ by $H_{\delta, n'}(X, V)$.
	For a subset of vertices $V'$ we let $H_{\delta, n'}(X, V')$ be the set of vertices that are $(\delta, n')$ heavy in the subgraph induced by $V'$.
\end{definition}

In other words, a vertex is $(\delta, n')$-heavy if it is close to $\delta^2n'$  other vertices in the given embedding. The next structural property of graphs is what will separate semirandom instances from worst-case instances.

\begin{definition}[Geometric expansion]\label{definition:geometric-expansion}
	A graph $G(V,E)$ on $n$ vertices satisfies the geometric expansion property at scale $(\delta, n', \alpha)$ if, for every feasible solution $X$ to \cref{eq:balanced-cut-sdp} on input $G$ and every subset $V'\subseteq V$ such that $H_{\delta, n'}(X, V') = \emptyset$, it holds
	\begin{align*}
		\Card{\Set{\ij \in E\cap (V'\times V')\suchthat \Snorm{v_i-v_j}\leq \delta}}\leq 10\cdot \delta^2 \cdot \alpha\,.
	\end{align*}
\end{definition}

That is,  a graph is geometrically expanding if the uniform distribution over the edges of non-heavy vertices satisfies a one-sided Chebyshev's inequality. 
For simplicity, we say that a graph $G$ is geometrically expanding up to scale $(100^{-z}, n, \alpha)$ if it is geometrically expanding at scale $(100^{-i}, n, \alpha)$ for all $1\leq i\leq z\,.$ 
We remark that \cref{definition:geometric-expansion} is equivalent to the geometric expansion property defined in \cite{DBLP:conf-stoc-MakarychevMV12}.

We are now ready to present the main theorem of the section.

\begin{theorem}[Main theorem]\label{theorem:balanced-cut-technical}
	There exists a randomized algorithm that on input $a, \alpha >\Omega(1)\,, \kappa \,,  \delta > 1/\log n$ and a graph $G$ such that:
	\begin{enumerate}
		\item there exists an  $a$-balanced partition $(A,B)$ with $\Card{E(A, B)}\leq \alpha\,,$
		\item $G(V, E(A, B))$ is a geometric expander up to scale $(100\delta, n, \alpha)$, 
	\end{enumerate}
	returns an $\Omega(a)$-balanced partition $(S,T)$ with cut $$\Card{E(S, T)}\leq O(\alpha ) (1+\delta \cdot \kappa\cdot \sqrt{\log n})$$ with probability $1-o(1)$.
	Moreover, the algorithm runs in time  $\tilde{O}\Paren{\Card{V(G)}^{1+O(1/\kappa^2)+o(1)}+\Card{E(G)}^{1+O(1/\kappa^2)+o(1)}}$. 
\end{theorem}

\cref{theorem:main} essentially follows from \cref{theorem:balanced-cut-technical} observing that graphs generated through \cref{model:main} are good geometric expanders.
To show this first observe that random bipartite graphs are good geometric expanders.

\begin{theorem}[Geometric expansion of random graphs, \cite{DBLP:conf-stoc-MakarychevMV12}]\label{theorem:geometric-expansion-random-graphs}
	Let $t>0$. Let ${G}$ be a graph over $n$ vertices generated through the first two steps (i), (ii) of \cref{model:main} with parameters $a, \eta>0\,.$ 
	Then, with probability $1-n^{-\Omega(1)}$, ${G}$ is geometrically expanding up to scale $\Paren{100^{-t}, n, \Theta(n^2\cdot \eta + 100^{t}\cdot n\cdot t^2)}$.
\end{theorem}

Second, observe that geometric expansion in  bipartite graphs is a property that is to some extent robust to changes monotone with respect to the bipartition.

\begin{fact}[Robustness of geometric expansion, \cite{DBLP:conf-stoc-MakarychevMV12}]\label{fact:robustness-geometric-expansion}
	Let ${G}$ be a graph over $n$ vertices generated through the first two steps (i), (ii) of \cref{model:main} and let $G^\circ$ be a graph obtained after applying steps (iii), (iv). If ${G}$ is geometrically expanding up to scale $\Paren{\delta, n, \tau}$, then so  is $G^\circ(V, E(A, B))$.
\end{fact}

This statement above implies that for $\eta\geq \Omega\Paren{\frac{(\log n)^2 \cdot (\log\log n)^2}{n}}$, with high probability $G^\circ(V, E(A, B))$ is a good geometric expander.
Now \cref{theorem:main} immediately follows combining \cref{theorem:balanced-cut-technical} \cref{theorem:geometric-expansion-random-graphs} and \cref{fact:robustness-geometric-expansion}.

\subsection{The algorithm}

We present here the algorithm behind \cref{theorem:balanced-cut-technical}. Since we will work using the matrix multiplicative framework (see \cref{section:background} for the necessary definitions), our main challenge is that of designing an appropriate oracle.
For simplicity, we split \oracle in three parts, the first two are due to  \cite{DBLP:conf-stoc-AroraK07,DBLP:conf-focs-Sherman09} the third part is our crucial addition and the main technical contribution of this work. Recall we denote by $\alpha$ the minimum objective of the program at hand.

\begin{lemma}[\cite{DBLP:conf-stoc-AroraK07}]\label{lemma:basic-oracle}
	Let $a>\Omega(1)$.
	There exists a  $\tilde{O}(\alpha/n)$-bounded, $\Theta(\log n)^2$-robust, $\Theta(1)$-separation oracle that, given a candidate solution to \cref{eq:balanced-cut-sdp}  with input graph $G$ on $n$ vertices and $a$-balanced cut of value at most $\alpha$, outputs \textbf{no} if one of the following conditions are violated. Let $W:=\Set{i\in [n]\,|\, \Snorm{v_i}>2} \subseteq [n]$ and $S:= [n]\setminus W$. 
	\begin{itemize}
	    \item Flatness: $\Card{W}< \frac{n}{(\log n)^{100}}\,.$
	    \item Balance: $\sum_{i,j \in S} \Snorm{v_i-v_j}\geq 2an$.
	\end{itemize}
	Moreover, the oracle is $\cT$-lean for some $\cT\leq O\Paren{(\card{V(G)}+\card{E(G)}}$
\end{lemma}

We omit the proof of \cref{lemma:basic-oracle} as it can be found in \cite{DBLP:conf-stoc-AroraK07}.
If both the flatness and the balance condition are satisfied, then we apply the following oracle,  due to \cite{DBLP:conf-focs-Sherman09}.

\begin{lemma}[\cite{DBLP:conf-focs-Sherman09}]\label{lemma:flow-oracle}
	Let $\kappa, a>0\,, 0< \delta<1/200, a >\Omega(1)$.
	There exists a  $\tilde{O}(\alpha/ n)$-bounded, $O(\log n)^{2}$-robust, $\Theta(1)$-separation oracle that, given a candidate solution to \cref{eq:canonical-balanced-cut-primal-sdp} with input graph $G$ on $n$ vertices and  $a$-balanced cut of value at most $\alpha$, 
	outputs \textbf{yes} only if it finds an $\Omega(a)$-balanced partition $(P, P')$ of $V(G)$ satisfying
	\begin{align*}
		\sum_{i \in P\,, j \in P'\,, \ij \in E(G)} \Snorm{v_i-v_j}&\leq O(\alpha)\\
		\Card{E(P, P')}&\leq O(\alpha\cdot \kappa)\cdot \sqrt{\log n}\,.
	\end{align*}
	Moreover, the oracle is $\cT$-lean for some $\cT\leq \tilde{O}\Paren{\Card{V(G)}^{1+O(1/\kappa^2)+o(1)}+\Card{E(G)}^{1+O(1/\kappa^2)+o(1)}}$.
\end{lemma}

The proof of \cref{lemma:flow-oracle} can be found in \cite{DBLP:conf-focs-Sherman09}. 
The improvement on the time complexity follows by \cref{theorem:max-flow-linear-time}. 
The next result is the \textit{crucial} addition we need to the oracle of \cite{DBLP:conf-stoc-AroraK07,DBLP:conf-focs-Sherman09}.  We prove it in \cref{section:heavy-vertices-removal}.

\begin{lemma}\label{lemma:oracle-heavy-vertices}
	Let $0<\ell\leq 1\,, \alpha, a >0$ and $0< \delta \leq 1/200$.
	Let $G$ be a graph on $\ell \cdot n$ vertices such that
	\begin{itemize}
		\item it has a $a$-balanced partition $(A, B)$ with $\Card{E(A, B)}\leq \alpha\,,$
		\item $G(V, E(A, B))$ is geometrically expanding up to scale $(\delta, n, \alpha)$. 
	\end{itemize}
	
	There exists a $\tilde{O}(\alpha/\ell\cdot n)$-bounded, $O(\log n)^{100}$-robust, $\Theta(1/\log n)$-separation oracle that, given a candidate solution to \cref{eq:canonical-balanced-cut-primal-sdp} with input graph $G$, 
	either outputs \textbf{no}, or outputs a set of edges $E^*\subseteq E(G)$ of cardinality $O(\alpha/\delta)$ and partition $(P_1, P_2, V')$ of $V(G)$ such that
	\begin{enumerate}
		\item 
		$\Card{E(P_1, P_2, V')\setminus E^*}\leq O\Paren{ \frac{\alpha}{\delta}\Paren{1+\frac{\ell}{\delta}}}\,.$
		\item  $\Abs{\Card{P_1} - \Card{P_2}} \leq a\cdot n/2\,.$ 
		\item 
		$\forall \ij \in E(G)\setminus E^*$ with $i,j \in V'$ it holds $\Snorm{v_i-v_j}\leq \delta\,.$
		\item $H_{\delta, n}(X, V')=\emptyset\,.$ 
	\end{enumerate}
	Moreover the oracle is $\cT$-lean for some $\cT\leq \tilde{O}\Paren{\Card{V(G)}+\Card{E(G)}}$ and $1-O(\log n)^{-50}$-reliable.
\end{lemma}

Before presenting the algorithm that uses the oracle above, let's briefly discuss its meaning.
Notice the \textit{heavy vertices} condition (4). This  ensures that in the subgraph $G(V', E\setminus E^*)$ any feasible embedding has weight at most $10\alpha\cdot \delta^2$ on the edges in $(E(A,B)\setminus E^*)\cap (V'\times V')$. In other words, after paying the edges in the cut of the partition $(P_1, P_2, V')$, geometric expansion of the underlying graph guarantees that the minimum objective value of \cref{eq:balanced-cut-sdp}  now decrease by a $10\delta^2$ factor. 

Next we present the algorithm behind \cref{theorem:balanced-cut-technical} and prove its correctness. We denote by \oracle a combination of the oracles in \cref{lemma:basic-oracle}, \cref{lemma:flow-oracle} and \cref{lemma:oracle-heavy-vertices} obtained applying them sequentially (in this order).

\begin{algorithm}[ht]
   \caption{Fast and robust algorithm for balanced cut}
   \label{algorithm:balanced-cut}
\begin{algorithmic}
   \STATE {\bfseries Input:}  A graph $G$ with minimum $a$-balanced cut of value at most $\alpha$, $T\,, \kappa, d, \delta>0$.
   \STATE
   \STATE Set $\super{\delta}{0}=1/200$.
    \FOR{$i=1$ {\bfseries to} $T$} 
    \STATE Let $\super{G}{i}$ be the current remaining graph with optimal cut value $\super{\alpha}{i}$ and $\Card{V(\super{G}{i})}=:\super{n}{i}$.
    \STATE Run the approximate matrix multiplicative weights algorithm (\cref{algorithm:approx-mmw}) for program \ref{eq:canonical-balanced-cut-primal-sdp} using  \oracle (with parameter $\super{\delta}{i}$).
    \STATE Let $W^*\in \R^{d, \times n}$ be the returned embedding, $\super{E}{i}$ the set of edges found and $(\super{P}{i}, \super{P'}{i})\,, $ $(\super{P_1}{i}\,, \super{P_2}{i}\,, \super{V'}{i})$ the partitions found by  \oracle.
    \STATE Remove the edges in $\super{E}{i}$.
    \IF {$\Card{E(\super{P}{i}, \super{P'}{i})}\leq O(\alpha\cdot(1 + \delta\cdot \kappa \cdot \sqrt{\log n}))$}
    \STATE Exit the loop.
    \ELSE 
    \STATE Remove vertices in $\super{P_1}{i}$ and in $\super{P_2}{i}$. Set $\super{\delta}{i+1}$ to $\super{\delta}{i}/100$.
    \ENDIF
    \ENDFOR
    \STATE Let $(P,P')$ be the bipartition found by  \oracle in its last iteration.
    \STATE Arbitrarily assign sets $\super{P}{1}_1\,, \super{P_2}{1}\,, \ldots\,, \super{P_1}{i}\,, \super{P_2}{i}$ removed in previous iterations to $P$ or $P'$, keeping the two sides $a$-balanced. 
    \STATE {\bfseries Return} the resulting bipartition.
\end{algorithmic}
\end{algorithm}

\begin{proof}[Proof of \cref{theorem:balanced-cut-technical}]
	We set $T$ such that $\super{\delta}{T}=100\delta$ where $\super{\delta}{0}=1/200$.
	By construction of \oracle and \cref{corollary:running-time-approx-mmw}
	the algorithm runs in time $\tilde{O}\Paren{\Card{V(G)}^{1+O(1/\kappa^2)+o(1)}+\Card{E(G)}^{1+O(1/\kappa^2)+o(1)}}$.
	
	

	Now consider a fixed iteration $i$, we assume $\super{\alpha}{0}=\alpha\,, \super{\ell}{0}=1$. Let $\super{\alpha}{i}$ be the cost of the minimum feasible solution on the remaining graph $\super{G}{i}$ on $\super{\ell}{i}\cdot n$ vertices. Let $\super{E}{i}$ be the set of edges removed at iteration $i$ and $(\super{P}{i}_1, \super{P}{i}_2, \super{V'}{i})$ the partition at iteration $i$.
	Notice that if  at some point $\super{\alpha}{i}\leq O(\alpha\cdot (1 + \delta\cdot \kappa \cdot \sqrt{\log n}))$ then the algorithm breaks the cycle and returns a balanced partition.
	
	So we may assume that at the current iteration $i$, $\super{\alpha}{i}\geq \omega(\alpha\cdot ((1 + \delta\cdot \kappa \cdot \sqrt{\log n}))$. 
	Now the result follows by showing that, at each step, it holds
	\begin{align}\label{eq:geometric-expansion}
		\super{\alpha}{i}\leq 10\cdot \super{\alpha}{i-1}\cdot \Paren{\super{\delta}{i-1}}^2\,.
	\end{align} 
	Indeed suppose the claim holds. By construction all the edges in the final cut are in
	\begin{align*}
		\Paren{\underset{i <T}{\bigcup} \super{E}{i}} \cup \Paren{\underset{i< T}{\bigcup} E(\super{P}{i}_1, \super{P}{i}_2, \super{V'}{i})\setminus \super{E}{i}} \cup E(\super{P}{T}, \super{P'}{T})\,.
	\end{align*}
	By \cref{eq:geometric-expansion} we can bound the first term as $\card{\underset{i\leq T}{\bigcup} \super{E}{i}}\leq O(\alpha)\,.$
	For the second term:
	\begin{align*}
		\Card{\underset{i\leq T}{\bigcup} E(\super{P}{i}_1, \super{P}{i}_2, \super{V'}{i})}&\leq 
		\sum_{i\leq T} O\Paren{\frac{\super{\alpha}{i}}{\super{\delta}{i}}\Paren{1+\frac{\super{\ell}{i}}{\super{\delta}{i}}}}\\
		&\leq O\Paren{\sum_{i\leq T} \frac{\super{\alpha}{i}}{\super{\delta}{i} } + \frac{\super{\alpha}{i}\cdot \super{\ell}{i}}{\Paren{\super{\delta}{i}}^2}}\\
		&\leq O\Paren{\sum_{i\leq T}\super{\alpha}{i} \Paren{\super{\delta}{i}+\super{\ell}{i+1}}}\\
		&\leq O\Paren{\sum_{i\leq T}\super{\alpha}{i}\cdot \Paren{\super{\delta}{i}+\super{\ell}{i}}}\\
		&\leq \alpha\cdot O\Paren{\sum_{i\leq T} \super{\delta}{i}+\super{\ell}{i}}\\
		&\leq O(\alpha)\,,
	\end{align*}
	where in the second step we used the inequalities
	\begin{align*}
		\frac{\super{\alpha}{i}}{\super{\delta}{i}}\leq 10^3\cdot \super{\alpha}{i-1}\cdot \super{\delta}{i-1}\,,&\\
		\frac{\super{\alpha}{i}\cdot \super{\ell}{i}}{\Paren{\super{\delta}{i}}^2}\leq  10^5\cdot \super{\alpha}{i-1}\cdot \super{\ell}{{i}}\,,&
	\end{align*}
	both following from \cref{eq:geometric-expansion}.
	For the third term  we have
	$\super{\alpha}{T}\leq O\Paren{\alpha\cdot \delta}$ by construction.
	Thus by  \cref{lemma:flow-oracle} we get $\Card{E(\super{P}{T}, \super{P'}{T})}\leq O\Paren{\alpha\cdot \kappa\cdot \delta \sqrt{\log n}}\,.$
	
	It remains to prove \cref{eq:geometric-expansion}.
	At each iteration $i$, the set $V(\super{G}{i})$ does not contain edges of length more than $\super{\delta}{i}$ in the embedding as well as $(\super{\delta}{i}, n)$ heavy vertices. Thus by \cref{definition:geometric-expansion}, the set $V(\super{G}{i})$ has a $\Omega(a)$-balanced cut of value $O(\super{\alpha}{i})$ . Then \cref{eq:geometric-expansion} follows as desired.
\end{proof}

\section*{Acknowledgments}
Tommaso d’Orsi is partially supported by the project MUR FARE2020 PAReCoDi.

\bibliography{custom}
\bibliographystyle{icml2024}

\newpage
\appendix

\section{Background}\label{section:background}
We introduce here background notion used throughout the paper.

\paragraph*{Maximum flow}
Let $G(V, E)$ be a graph.
For a flow which assigns  value $f_p$ to path $p$ define $f_e$ to be the flow on edge $e\in E(G)$, i.e. $f_e :=\sum_{p\ni e} f_p$. Define  $f_\ij$ to be the total flow between nodes $i\,, j$,  i.e. $f_\ij = \sum_{p\in P_\ij}f_p$, where $P_\ij$ is the set of paths from $i$ to $j$. Similarly, define $f_i$ to be flow from node $i$. That is, $f_i=\sum_{j\in[n]}f_{ij}\,.$ A valid $d$-regular flow is one that satisfies the capacity constraints: $\forall e\in E\,:\, f_e\leq 1$ and $\forall i \in V\,:\, f_i\leq d$.
For a partition $(A, B)$ of $G$, the maximum $d$-regular flow between $A$ and $B$ is the  maximum $d$-regular flow between vertices $s$ and $t$ in the graph obtained from $G$ as follows: (1) connect all vertices in $A$ to a new vertex $s$ by edges of capacity $d$, (2) connect all vertices in $B$ to a new vertex $t$ by edges of capacity $d$.
 
Through the paper we always assume the capacities $d$ to be integral and bounded by $O(\poly(n))$. We assume the algorithm used to compute the maximum flow is the near linear time algorithm in \cite{DBLP:journals-corr-abs-2203-00671}, captured by the result below:
 
 \begin{theorem}[Maximum flow in almost linear time \cite{DBLP:journals-corr-abs-2203-00671}]\label{theorem:max-flow-linear-time}
 	Let $G$ be a graph on $n$ vertices and let $d$ be  integral of value at most $O(\poly(n))$. There exists an algorithm computing the maximum $d$-regular flow between two vertices in time at most $O(\Card{E(G)}^{1+o(1)})$.
 \end{theorem}

\subsection{The matrix multiplicative weights method for SDPs}\label{mmw-background}
We  recall here how the matrix multiplicative method can be used to approximately solve semidefinite programs.  \cite{DBLP:conf-stoc-AroraK07, DBLP:conf-soda-Steurer10}. As most of the notions presented here already appeared in \cite{DBLP:conf-stoc-AroraK07, DBLP:conf-soda-Steurer10}, we encourage the knowledgeable reader to skip this section, move directly to \cref{section:mmw-algorithm} and come back when needed.
We focus on minimization problems although the same framework applies to maximization problems.

A primal semidefinite program over $n^2$ variables (i.e. the $n$-by-$n$ matrix variable $X$) and $m$ constraints can be written in its canonical form as
\begin{equation}\label{eq:canonical-primal-sdp}
	\Set{
		\begin{aligned}
		&\min &\iprod{L, X}\\
		&\forall j \in [m]\,, &\iprod{A_j, X}\geq b_j\\
		&&X\sge \mathbf{0}
		\end{aligned}
	}
\end{equation} 
Here $A_1,\ldots, A_m, L$ are symmetric matrices. 
We denote the feasible set of solutions by $\cX$ and the optimal objective value by $\alpha$. For simplicity we assume  that $A_1=-\Id_n$ and $b_1 =-r$. This serves to bound the trace of the solution so that $\cX\subseteq \Delta_n(r)$. The associated dual, with variables $y_1,\ldots, y_m$, is the following program

\begin{equation}\label{eq:canonical-dual-sdp}
	\Set{
		\begin{aligned}
		&\max &\iprod{b, y}\\
		&&\sum_{j\in [m]} A_jy_j\sle L\\
		&&y\geq \mathbf{0}
		\end{aligned}
	}
\end{equation} 
where $b$ is the $m$-dimensional vector with entries $b_1,\ldots,b_m$.

For a convex set $\cX^*\subseteq\Delta_n(r)$ (think of $\cX^*$ as the set of feasible solution to a program of the form \cref{eq:canonical-primal-sdp} with objective value close to the optimum) a $\gamma$-\textit{separation} \oracle is an algorithm that , given a candidate matrix $X$, outputs one of the following:

\begin{itemize}
	\item \textbf{yes}:  the \oracle determines $X$ is "close" (the precise notion of closeness is problem dependent) to $\cX^*$. 
	\item \textbf{no}: the \oracle finds a hyperplane that separates $X$ from $\cX^*$ by a $\gamma$-margin. That is, it outputs a symmetric matrix $M$ such that for all $X'\in \cX^*$ we have $\iprod{M,X'}\geq 0$ while $\iprod{M, X} < -\gamma \alpha$.
\end{itemize}

A $\gamma$-separation \oracle is said to be $\zeta$-bounded if $\Norm{M}\leq \zeta$ for any hyperplane $M$ found by the \oracle.
The boundedness of the \oracle will be relevant for the running time of our algorithms. It is important to notice that the parameters $\zeta, \gamma$ are not independent, in particular one may increase $\gamma$ by scaling up the corresponding matrix $M$. We keep them distinct for  convenience .

Concretely, given a program of the form \cref{eq:canonical-primal-sdp} and a candidate solution $X$, we will consider \oracle algorithms that, in the \textbf{no} case, find a pair $(y, F)$ where $F$ is a matrix in $\cS_n$ satisfying $F \sle L$ and $y$ is a candidate solution\footnote{Not necessarily feasible.}  for the dual program \cref{eq:canonical-dual-sdp} such that
$y\in \Set{y \suchthat \iprod{b, y}\geq \alpha\,, y\geq 0}$ and 
\begin{align*}
	\iprod{\sum_{j\in [m]}A_j y_j-F, X} \leq - \gamma\cdot \alpha\,.
\end{align*}
It is easy to see that this is indeed a separating hyperplane as for any feasible solution $X'$ with objective value less than $\alpha(1+\gamma)$
\begin{align*}
	\iprod{\sum_{j\in [m]}A_jy_j-F, X'}\geq \sum_{j\in [m]}b_jy_j- \iprod{L, X'}> \alpha-(1+\gamma)\alpha = -\gamma\cdot  \alpha \\,.
\end{align*}

We will use our oracle algorithms in the following  framework. 

\begin{algorithm}[ht]
   \caption{Matrix multiplicative weights algorithm for SDPs}
   \label{algorithm:mmw}
\begin{algorithmic}
    \STATE {\bfseries Input:}  A program of the form \cref{eq:canonical-primal-sdp} with optimal value $\alpha$, a $\zeta$-bounded $\gamma$-separation \oracle, parameters $T,\epsilon, r$.
	\STATE
    \STATE Set $\super{X}{1} = \frac{r}{n}\Id_n$.
    \FOR{$t=1$ {\bfseries to} $T$} 
    \STATE Run the \oracle with candidate solution $\super{X}{t}$.
    \IF{the \oracle outputs \textbf{yes}}
    \STATE {\bfseries Return} $\super{X}{t}$.
    \ELSE
    \STATE Let $(\super{y}{t}, F)$ be the pair generated by \oracle. Set $\super{Y}{t} = \Paren{\sum_{j\in [m]}A_j\super{y_j}{t}-F+\zeta\Id_n}/2\zeta$.
    \ENDIF
    \STATE Compute $\super{X}{t+1} = r\cdot \exp\Paren{\eps\sum_{t'\leq t}\super{Y}{t'}}/\Tr\exp\Paren{\eps\sum_{t'\leq t}\super{Y}{t'}}$
    \ENDFOR
\end{algorithmic}
\end{algorithm}

The choice of the iterative updates in step 4 is based on the matrix multiplicative weights method. In particular, this allows one to obtain the following crucial statement.

\begin{theorem}[\cite{DBLP:conf-stoc-AroraK07}]\label{theorem:result-mmw}
	Consider \cref{algorithm:mmw}.
	Let $\eps\leq \gamma\alpha/(2\zeta\cdot r)$ and $T\geq 2\eps^{-2} \log n$. If there exists a feasible solution with value at most $\alpha(1+\gamma)$, then \oracle will output \textbf{yes} within $T$ iterations.
\end{theorem}

\subsubsection{Approximate matrix exponentiation, robust and reliable oracles}\label{section:approximate-matrix-exponentiation}

There are two issues with \cref{theorem:result-mmw} if one aims for near linear running time: first, already writing down $\super{X}{t}$ requires time quadratic in $n$; second, algorithms known to compute the matrix exponentiation are slow . One can circumvent these obstacles computing the exponentiation only approximately while also keeping only an approximate representation of $\super{X}{t}$.
To formalize this we introduce additional notation. For a positive semidefinite $n$-by-$n$ matrix $M$, we let $P_{\leq p}(M)$ be the degree-$p$ approximation of the matrix exponential $\exp(X)$:
\begin{align*}
	P_{\leq p}(M):= \sum_{i\leq p}\frac{1}{i!} M^i\,.
\end{align*}
Recall that for a matrix $M$, the Gram decomposition of the exponential $\exp(M)$ is  $\exp(M)=\transpose{\exp(\tfrac{1}{2}M)}\exp(\tfrac{1}{2}M)$ thus we may see  $P_{\leq p}(\tfrac{\eps}{2}\sum_{t'\leq t}\super{Y}{t'})$ as a a matrix having as columns low-degree approximations of the Gram vectors of $\exp(M)$.
One can then embed these vectors in a low dimensional space, without distorting their pair-wise distance by projecting them onto a random $d$-dimensional subspace:

\begin{lemma}[\cite{johnson1984extensions}]\label{lemma:jl}
	Let $\Phi$ be a $d$-by-$n$ Gaussian matrix, with each entry independently chosen from $N(0, 1/d)$. Then, for every vector $u \in \R^n$ and every $\eps \in(0,1)$
	\begin{align*}
		\bbP\Paren{\Norm{\Phi u} = \Paren{1\pm \eps}\Norm{u}}\geq 1-e^{-\Omega(\eps^2 d)}\,.
	\end{align*}
\end{lemma}

We will follow this strategy to speed up \cref{algorithm:mmw}.

\begin{algorithm}[ht]
   \caption{Approximate matrix multiplicative weights algorithm for SDPs}
   \label{algorithm:approx-mmw}
\begin{algorithmic}
    \STATE {\bfseries Input:}  A program of the form \cref{eq:canonical-primal-sdp} with optimal value $\alpha$, a $\zeta$-bounded $\gamma$-separation \oracle, parameters $T,\epsilon, r, d, p$, a $d$-by-$n$ random matrix  $\Phi$ with  i.i.d entries from $N(0, 1/d)$. 
    \STATE
    \STATE Set $\super{W}{1} = \frac{r}{n}(\Phi\Id_n)/\Tr(\Phi\Id_n)$.
    \FOR{$t=1$ {\bfseries to} $T$} 
    \STATE Run the \oracle with candidate solution $\super{W}{t}$.
    \IF{the \oracle outputs \textbf{yes}}
    \STATE {\bfseries Return} $\super{W}{t}$.
    \ELSE
    \STATE Let $(\super{y}{t}, F)$ be the pair generated by \oracle. Set $\super{Y}{t} = \Paren{\sum_{j\in [m]}A_j\super{y_j}{t}-F+\zeta\Id_n}/2\zeta$.
    \ENDIF
    \STATE Sample a $d$-by-$n$ random matrix $\Phi$ with  i.i.d entries from $N(0, 1/d)$.
    \STATE Compute $\super{W}{t} = r\cdot \Phi P_{\leq p}(\tfrac{\eps}{2}\sum_{t'\leq t}\super{Y}{t'}) /\Tr \Paren{\Phi P_{\leq p}(\tfrac{\eps}{2}\sum_{t'\leq t}\super{Y}{t'})}$.
    \ENDFOR
\end{algorithmic}
\end{algorithm}

Observe that $P_{\leq p}(\tfrac{\eps}{2}\sum_{t'\leq t}\super{Y}{t'}) $ corresponds to a low degree approximation of the Gram vectors of the matrix $\super{X}{t}$ in step 4 of \cref{algorithm:mmw}. We then compute $\super{W}{t}$ by embedding these vectors in a random $d$-dimensional space.

The   statement below shows that in many cases we can compute such matrices $\super{W}{t}$ very efficiently.

\begin{lemma}[\cite{DBLP:conf-soda-Steurer10}]\label{lemma:fast-computation-approx-mmw}
	Suppose we can perform matrix-vector multiplication with the matrices $\super{Y}{t}$ in time $\cT$. Then, for every $t$, we can compute $\super{W}{t}$ in time $O(t\cdot  p \cdot d \cdot \cT)$.
\end{lemma}

A priori it is not clear whether \cref{algorithm:approx-mmw} can provide the same guarantees of \cref{algorithm:mmw}. However, the next result show this is the case under reasonable circumstances. 

\begin{definition}[$d$-robust oracle, extension of  \cite{DBLP:conf-soda-Steurer10}]\label{definition:robust-oracle}
	We say that a $\zeta$-bounded $\gamma$-separation oracle is $d$-robust if for every matrix $X\in  \Delta(r)$ with $X = \transpose{W}W$ 
	\begin{align*}
		\bbP_{\Phi\sim N(0,1/d)^{d\times n}}&\Paren{\text{\oracle outputs \textbf{no} on input $\transpose{(\Phi W)}\Phi W$ and $\iprod{\super{Y}{t}, X}\geq -\frac{3}{4}\gamma\alpha$}}\\&\leq \frac{(\gamma\alpha /\zeta r)^2}{(\log n)^{10}}\,.
	\end{align*}
\end{definition}

\begin{lemma}[\cite{DBLP:conf-soda-Steurer10}]\label{lemma:robust-oracle}
	Consider \cref{algorithm:approx-mmw}. Let $\eps\leq \gamma\alpha/(2\zeta\cdot r)$, $T\geq 2\eps^{-2} \log n$ and $p\geq 10\eps^{-1}\log n$. Suppose we have a $d$-robust $\zeta$-bounded $\gamma$-separation \oracle.
	If there exists a feasible solution with value at most $\alpha(1+2\gamma)$, then \oracle will output \textbf{yes} within $T$ iterations with probability at least $1-O(\log n)^{-10}$.
\end{lemma}

We can combine \cref{lemma:robust-oracle}, \cref{lemma:fast-computation-approx-mmw} and \cref{lemma:fast-matrix-vector-multiplication} to obtain a user-friendly statement concerning the running time of \cref{algorithm:approx-mmw}.
We introduce two additional definitions.

\begin{definition}[$\cT$-lean oracle]\label{definition:lean-oracle}
	We say that a $\zeta$-bounded $\gamma$-separation $d$-robust is $\cT$-lean if:
	\begin{itemize}
		\item the oracle compute its outputs in time at most $O(\cT)$.
		\item If the oracle outputs $\textbf{no}$, the matrix-vector multiplication between an arbitrary vector and the feedback matrix $\Paren{\sum_{j\in [m]}A_j y_j-F+\zeta \Id_n}/2\zeta$ can be computed in time $O(\cT)$.
	\end{itemize}
\end{definition}

We remark that one can upper bound the time needed for matrix-vector multiplication by the number of non-zero entries in the matrix of interest.

\begin{fact}\label{lemma:fast-matrix-vector-multiplication}
	Let $M\in \R^{n\times n}$ be a matrix with $m$ non-zero entries and let $v\in \R^{n}$. There exists an algorithm that computes $Mv$ in time $O(m+n)$.
\end{fact}

The next definition formalizes the idea of oracles that may find a separating hyperplane only with certain probability. 

\begin{definition}[$q$-reliable]\label{definition:q-reliable}
	We say that a $\zeta$-bounded, $\gamma$-separation, $d$-robust, $\cT$-lean oracle is $q$-reliable if the probability (over random bits) that it outputs \textbf{no}  for any feasible solution with objective value at most $(1+2\gamma)\alpha$ is at most $1-q$.
\end{definition}

For oracles that are $1$-reliable we omit mentioning their reliability.
We are ready to present a user-friendly running time statement, which we will use as a black box. 

\begin{corollary}[Running time of \cref{algorithm:approx-mmw}]\label{corollary:running-time-approx-mmw}
	Let $\oracle$ be a $\zeta$-bounded, $\gamma$-separation, $d$-robust, $\cT$-lean  $q$-reliable oracle.  Then, for $\eps\leq \gamma\alpha/(2\zeta\cdot r)$, $T\geq 2\eps^{-2} \log n$ and $p\geq 10\eps^{-1}\log n$, with probability at least $1-O(\log n)^{-10}-(1-q)T$ over random bits, \cref{algorithm:approx-mmw} terminates in time $O\Paren{T^2\cdot \cT\cdot d\cdot p}\,.$
	\begin{proof}
		The Corollary follows immediately from \cref{lemma:fast-computation-approx-mmw}, \cref{lemma:robust-oracle}, \cref{definition:lean-oracle} and \cref{definition:q-reliable}.
	\end{proof}
\end{corollary}

\section{The heavy vertices removal oracle}\label{section:oracle}

We prove here \cref{lemma:oracle-heavy-vertices}.
In \cref{section:heavy-vertices-removal} we introduce a procedure that the oracle uses to find either the partition or a separating hyperplane. Then in \cref{section:bounded-oracle-heavy-vertices} we prove the Lemma.
Throughout the section we consider the following parameters range:
\begin{align}\label{eq:parameters}
	n, \alpha> 0\,, \Omega(1)\leq a\leq 1\,, 0< \ell \leq 1\,,  \Omega(1\log n)\leq \delta \leq 1/200\,.
\end{align}

\subsection{The fast heavy vertices removal procedure}\label{section:heavy-vertices-removal}

We introduce  the main procedure used by  \oracle.
The central tool of the section is the following statement. 

\begin{lemma}\label{lemma:result-heavy-vertices-removal-procedure}
	Consider the parameter settings of \cref{eq:parameters}.
	Let $G$ be a graph on $\ell \cdot n$ vertices with $a$-balanced cut of value at most $\alpha$ that is geometrically expanding up to scale $(\delta , n, \alpha)$. 
	
	Let $X$ be a feasible solution for \cref{eq:balanced-cut-sdp}  on input $G$, with objective value $O(\alpha)$.
	There exists a randomized procedure (\cref{algorithm:heavy-vertex-removal-procedure}) that outputs a set of edges $E^*\subseteq E(G)$ of cardinality $O(\alpha/\delta)$ and  a partition $(P_1, P_2, V')$ of $V(G)$ satisfying $(2)\,, (3)\,, (4)$ in \cref{lemma:oracle-heavy-vertices} and such that
	\begin{align*}
		\E \Brac{\Card{E(P_1, P_2, V')\setminus E^*}}\leq C\cdot\frac{\alpha}{\delta} \Paren{1+\frac{\ell}{\delta }}\,,
	\end{align*}
	where $C>0$ is a universal constant.	
	Moreover, if the solution is given in the form of $v_1,\ldots, v_n\in \R^{O(\polylog n)}$,  the procedure runs in time $\tilde{O}( \Card{V(G)}+\Card{E(G)})$.
\end{lemma}

The first building block towards a proof of \cref{lemma:result-heavy-vertices-removal-procedure} is the result below, which introduces a subroutine to identify heavy vertices.

\begin{lemma}\label{lemma:procedure-find-heavy-vertices}
  Consider the settings of \cref{lemma:result-heavy-vertices-removal-procedure}. Let $\rho \ge 2$.
  There exists a randomized procedure that
  outputs with probability at least $1-1/n$ a set of vertices $V^*$ and a mapping $f:V\rightarrow V^*\cup \Set{(*)}$ such that
  \begin{enumerate}
  \item Each vertex $i$ of $V^*$ satisfies $\Card{\Set{j\in V \suchthat \Snorm{v_i-v_j}\leq \rho\delta}}\geq 10\delta^2n$; and
  \item The set $W := \{i \mid f(i) = (*) \}$ does not contain a vertex $i$ such that
    $\Card{\Set{j\in V \suchthat \Snorm{v_i-v_j}\leq \delta}}\geq 10\delta^2n$.   
  \item $f(i)= j$ if there exists some $j\in V^*$ with $\Snorm{v_i-v_j}\leq \rho\delta\,,$
  \item $f(i)=(*)$ otherwise.
  \end{enumerate}  
  Moreover, if the solution is given in the form of $v_1,\ldots, v_n\in \R^{O(\polylog n)}$,
  the procedure runs in time $O\Paren{\frac{1}{\delta^2} \cdot \Card{V(G)} \polylog n}$. 
\end{lemma}
\begin{proof}
  We propose and analyze the following algorithm:
  \begin{enumerate}
  \item $S \gets$  Sample $100\delta^{-2} \log n$ points uniformly at random.
  \item For each point $i \in S$, compute $N(v)= \Set{j\in W \suchthat \Snorm{v_i-v_j}\leq \rho\delta}$.
  \item $V^* \gets S \setminus \Set{i \in S \suchthat \Card{N(v)} < 10\delta^2 n}$
  \item For each vertex $i$, if there exists $j \in V^*$ such that $\Snorm{v_i-v_j}\leq \rho\delta$ then
    $f(i) = j$ otherwise $f(i) = (*)$.
  \end{enumerate}
  Clearly the above procedure runs in time $O(\Card{V(G)} \cdot \Card{S}\cdot \polylog n)$ as desired, where the
  bulk of the work is done in the second and fourth steps.
  By the definition of the procedure, the first, third and fourth bullets of the theorem statement are satisfied. We thus need to
  show that the set $W := \{i \mid f(i) = (*) \}$ does not contain a vertex $i$ such that
    $\Card{\Set{j\in W \suchthat \Snorm{v_i-v_j}\leq \delta}}\geq 10\delta^2n$.
  
  A simple coupon collector argument implies that with probability at least $1-1/n^2$, for each
  vertex $i$, if $\Card{\Set{j\in V(G) \suchthat \Snorm{v_i-v_j}\leq \delta}}\geq 10\delta^2n$, then
  $\Set{j\in V(G) \suchthat \Snorm{v_i-v_j}\leq \delta} \cap S \neq \emptyset$. Thus, let $j^*$ be a vertex in
  $\Set{j\in V(G) \suchthat \Snorm{v_i-v_j}\leq \delta} \cap S \neq \emptyset$. Then, since $\rho \ge 2$, we have that
  $\Card{\Set{j\in V(G) \suchthat \Snorm{v_{j^*}-v_j}\leq \rho\delta}} \ge 10\delta^2 n$ and so $j^* \in V^*$ and
  $f(i) \neq (*)$ as desired. It remains to take a union bound over the probability of failure for each individual vertex,
  and we conclude that the overall failure probability is at most $1/n$.

\end{proof}  

We use the procedure in \cref{lemma:procedure-find-heavy-vertices} as a subroutine of the one presented next, which  for feasible embeddings finds a paritition satisfying $(2)\,, (3)\,, (4)$ in \cref{lemma:oracle-heavy-vertices} and $(1)$ in expectation.

\begin{algorithm}[ht]
   \caption{Fast heavy vertex removal procedure}
   \label{algorithm:heavy-vertex-removal-procedure}
\begin{algorithmic}
    \STATE {\bfseries Input:} A graph $G$ on $\ell\cdot n$ vertices, a candidate solution $X$ to
    \cref{eq:balanced-cut-sdp} on input $G$, parameters $a, \delta>0 \,, C >200\,.$
    \STATE
    \STATE Remove all edges of length at least $\delta$ in the embedding. Let $E^*$ be the set of such edges.   
    \STATE Find the set  $V^*$ via the subroutine in \cref{lemma:procedure-find-heavy-vertices} with $\rho=2$.
    \STATE Pick a maximal set $U$ of vertices in $V^*$ at pairwise squared distance at least $10\cdot C\delta$ in the embedding.
    \IF{$\Card{U}\geq \frac{a}{C\cdot \delta}\,:$}
    \STATE Pick $r\overset{u.a.r.}{\sim}[1, 2]$.
    \STATE For each $i\in U$, remove $i$ and all vertices at distance $\leq 2r\cdot \delta$ in the embedding. Let $U_i$ be the set of removed vertices via $i$. 
    \STATE {\bfseries Repeat} the algorithm on the remaining graph.
    \ELSE
    \STATE Run the subroutine \cref{algorithm:heavy-vertex-removal-last-iteration-subroutine} on the remaining graph and obtain additional sets $U_i$'s.
    \ENDIF
    \STATE  Distribute evenly the vertices in the $U_i$'s among two sets $P_1, P_2$ so that if $j,k \in U_i $ then $j,k $ are in the same set. Let $V' = V\setminus(P_1\cup P_2)\,.$
    \STATE {\bfseries Return} the partition $(P_1, P_2, V')\,.$
\end{algorithmic}
\end{algorithm}

\begin{fact}\label{fact:running-time-heavy-vertex-removal-procedure}
	\cref{algorithm:heavy-vertex-removal-procedure} runs in time $\tilde{O}\Paren{ \Card{V(G)}+ \Card{E(G)}}\,.$
	\begin{proof}
	Step 1 requires $O(E)$ time.
	The steps 2-4 can be repeated at most $C\cdot\ell/(a\cdot \delta)$ times. Indeed no vertex can be in both $U_i$ and $U_{i'}$ at the same time (even if $X$ does not satisfy the triangle inequality constraints) and since by definition each $U_i$ contains at least $10\delta^2 \cdot n$ vertices, in $ C\cdot \ell/(a \cdot \delta)$ iterations we will have removed all vertices form the graph.
	For each of these iterations, step 2 requires time $\tilde{O}\Paren{\frac{1}{\delta^2} \Card{V(G)}}$ and step 3 requires time $O(\Card{V(G)}\cdot \poly(1/a \delta)\,.$ Step 4 runs in time at most $O(\Card{V(G)}/a\delta)\,.$
	
	As we show in \cref{fact:running-time-heavy-vertex-removal-last-iteration-subroutine}, \cref{algorithm:heavy-vertex-removal-last-iteration-subroutine} also runs in time $\tilde{O}\Paren{\Card{V(G)}+\Card{E(G)}}$.
	Step 6 can be done in time $\tilde{O}(\Card{V(G)})$  after ordering the sets $U_i$'s.
	Thus the statement follows as $\delta\geq 1/\poly\log(n)\,, a\geq \Omega(1)\,.$
	\end{proof}
\end{fact}

The subroutine of step 5 in \cref{algorithm:heavy-vertex-removal-procedure} is presented below.
\begin{algorithm}[ht]
   \caption{Subroutine of fast heavy vertex removal procedure}
   \label{algorithm:heavy-vertex-removal-last-iteration-subroutine}
\begin{algorithmic}
    \STATE {\bfseries Input:} A graph $G$ on $\ell\cdot n$ vertices, a candidate solution $X$ to \cref{eq:balanced-cut-sdp} on input $G$, the list of vertices $V^*$, parameters $a, \delta>0 \,.$
    \STATE
    \STATE Consider the graph $G^*(V^*, \emptyset)$.
    \FOR {$ij \in E(G)$}
    \IF {$\Snorm{v_i-v_j}\leq 200\delta$ and $f(i)\neq (*)\,, f(j)\neq (*)$}
    \STATE Connect $f(i)$ to $f(j)$ in $G^*$ (excluding self-loops).
    \ENDIF
    \ENDFOR
    \STATE Pick $r\overset{u.a.r.}{\sim}[1, 2]$ and for each connected component $U$ in $G^*$, remove all vertices at distance $\leq 2r\cdot \delta$ to some vertex $i\in U$. Index the resulting sets by arbitrary representative vertices in each component.
    \STATE {\bfseries Return} the resulting sets $U_i$'s.
\end{algorithmic}
\end{algorithm}

\begin{fact}\label{fact:running-time-heavy-vertex-removal-last-iteration-subroutine}
	\cref{algorithm:heavy-vertex-removal-last-iteration-subroutine} runs in time $\tilde{O}\Paren{\Card{V(G)}+\Card{E(G)}}\,.$
	\begin{proof}
		We use the mapping of \cref{lemma:procedure-find-heavy-vertices}. We can then construct the graph $G^*$ in time $O(\Card{E(G)})$. Moreover, notice that $\Card{E(G^*)}\leq \Card{E(G)}\,.$ We can find the connected components in $G^*$ in time $O(\Card{E(G^*)}+\Card{V(G^*)})$ and partition the vertices in $G$ according to such connected components in time $\tilde{O}(\Card{V(G)})$. The result follows.
	\end{proof}
\end{fact}

Next we bound the probability that an edge gets cut in \cref{algorithm:heavy-vertex-removal-procedure}.

\begin{lemma}\label{lemma:heavy-vertex-removal-procedure-probability-cut}
	Consider the settings of \cref{lemma:result-heavy-vertices-removal-procedure}. 
	At each iteration of steps 2-4 in \cref{algorithm:heavy-vertex-removal-procedure} as well as the one using \cref{algorithm:heavy-vertex-removal-last-iteration-subroutine} the following holds:
	\begin{align*}
		\forall i \, \text{ s. t. }\, f_i\neq (*)\quad \bbP \Paren{\exists U_k\in U\text{ s. t. } i\in U_k\,, j \notin U_k}\leq \frac{\Snorm{v_i-v_j}}{ 2\cdot \delta}\,.
	\end{align*}
	\begin{proof}
		Consider first an iteration of steps 1-3 in \cref{algorithm:heavy-vertex-removal-procedure}.
		By construction each vertex $i$ can be in at most one set $U_k$.
		Since $r$  is chosen uniformly at random in the interval $[1, 2]$ the claim follows.
		So consider  \cref{algorithm:heavy-vertex-removal-last-iteration-subroutine}.
		Again, by construction each vertex $i$ can be in at most one set $U_k$ so by choice of $r$ the inequality holds.
	\end{proof}
\end{lemma}

Now \cref{lemma:result-heavy-vertices-removal-procedure} follows as a simple corollary.

\begin{proof}[Proof of \cref{lemma:result-heavy-vertices-removal-procedure}]
	By \cref{fact:running-time-heavy-vertex-removal-procedure}, steps 1-3 in \cref{algorithm:heavy-vertex-removal-procedure} are repeated at most $C\cdot \ell/\delta$ times while \cref{algorithm:heavy-vertex-removal-last-iteration-subroutine} runs only once. By \cref{lemma:heavy-vertex-removal-procedure-probability-cut} we then have $\forall \ij \in E(G)\setminus E^*$
	\begin{align*}
		\bbP\Paren{\ij\in E(P_1, P_2, V')\setminus E^*} \leq C^* \frac{\Snorm{v_i-v_j}}{ \delta}\cdot \Paren{1+C\cdot \ell/ \delta}\,,
	\end{align*} 
	for some $C^*>0\,.$
	Then the bound on $\E\Brac{\Card{E(P_1, P_2, V')\setminus E^*}}$ follows by linearity of expectation.
	By definition of $V^*$, the set $V'$ does not contain $(\delta, n)$-heavy vertices as well as no edges of length at least $\delta$. so it satisfies conditions (3), (4) in \cref{lemma:oracle-heavy-vertices} . Finally condition (2) follows by the spreadness condition of feasible solutions.
\end{proof}

\subsection{The  oracle}\label{section:bounded-oracle-heavy-vertices}

We prove here \cref{lemma:oracle-heavy-vertices}. 
To simplify the description of the oracle, as in \cite{DBLP:conf-stoc-AroraK07}, we consider the following modification of \cref{eq:canonical-balanced-cut-primal-sdp}, which contains additional constraints. The two programs are equivalent as these constraints are \textit{implied} by the ones in \cref{eq:canonical-balanced-cut-primal-sdp}. 

\begin{equation}\label{eq:extended-balanced-cut-primal-sdp}
	\Set{
		\begin{aligned}
			\min &\iprod{L,X}\\
			&X_{ii}=1&\forall i \in [n] & \quad(\text{unit norm})\\
			&\iprod{T_p, X}\geq 0&\forall \text{paths $p$} &\quad(\text{triangle inequality})\\
			&\iprod{K, X}\geq  4an^2&&\quad(\text{balance})
		\end{aligned}
	}
\end{equation}

\begin{remark}
	We remark that, as we need not to explicitly write down the program, but only efficiently find separating hyperplanes, the use of \cref{eq:extended-balanced-cut-primal-sdp} does not imply an increase in the running time of the overall algorithm.
\end{remark}

We consider  the dual program of \cref{eq:extended-balanced-cut-primal-sdp}, which has variables $x_1,\ldots,x_n$ for each vertex, $f_p$ for every path $p$ and an additional variable $z$ for the  set $[n]$ considered in the primal.
\begin{equation}\label{eq:extended-balanced-cut-dual-sdp}
	\Set{
		\begin{aligned}
			\max &\sum_{i\in [n]}x_i + an^2z\\
			&\diag(x)+\sum_p f_pT_p+zK\sle L\\
			&f_p, z\geq 0 &\forall \text{ paths $p$}
		\end{aligned}
	}
\end{equation} 

Now, given a candidate solution $X$, our starting point is the procedure of \cref{lemma:result-heavy-vertices-removal-procedure}, which we use to remove vertices that are heavy in the current embedding. 

\begin{proof}[Proof of  \cref{lemma:oracle-heavy-vertices}]
	Throughout the proof, whenever we write a feedback matrix, all the variables that are not specified are set to $0$.
	Let $X$ be the matrix denoting the current embedding.
	We may assume without loss of generality that both the oracles in \cref{lemma:basic-oracle} and \cref{lemma:flow-oracle} outputted \textbf{yes} on $X$.
	We claim there are at most $\bar{C}\cdot \alpha/\delta$ edges of length at least $\delta$ in the embedding for some large enough constant $\bar{C}>0$.
	Suppose this is not the case, consider the following procedure: 
	\begin{itemize}
			\item Pick uniformly at random an $a$-balanced bipartition $A,B$.
			\item Compute the max $d$-regular $A$-$B$ flow with $d=O(\alpha)/n\,.$
	\end{itemize}
	In expectation the flow is larger than $\bar{C}\alpha$ (if it is smaller we have found a cut).
	If the flow is larger than $\bar{C}\alpha$ then let $F$ be the Laplacian of the flow graph and let $D$ be the Laplacian of the complete weighted graph where only edges $\ij$ with $i \in A\,, j \in B$  have weight $f_\ij$, and the rest have $0$ weight .
	Then by definition $\sum_p f_p T_p= F-D\,.$
	Thus we may set $x_i=\alpha/\Card{V(G)}$ for all $i \in V(G)$, $f_p$ as in the computed flow for all $p$ and all other variables to $0$. The feedback matrix $Y$ becomes $\frac{\alpha}{\Card{V(G)}}\Id+F-D-F=\frac{\alpha}{\Card{V(G)}}\Id-D$ and we have
	\begin{align*}
		\iprod{\frac{\alpha}{\Card{V(G)}}\Id-D, X}\leq \alpha -(\bar{C}-1)\cdot \alpha<  -\alpha<0\,.
	\end{align*}
	Notice also that $\Norm{\frac{\alpha}{\Card{V(G)}}\Id-D}\leq O(\alpha/\Card{V(G)})\,.$
	We repeat this procedure $O(\log n)^{100}$ times, by Markov's inequality the claim follows with probability at least $1-O(\log n)^{-99}\,.$
	Let $E^*\subseteq E(G)$ be the set of  edges of length at least $\delta$ in the embedding. Notice that we must have $\Card{E^*}\leq O(\alpha/\delta)\,.$
	
	Now we run the heavy vertex removal procedure \cref{algorithm:heavy-vertex-removal-procedure}. Let $(P_1, P_2, V')$ be the resulting partition.
	If the partition satisfies $(1), (2), (3), (4)$ in \cref{lemma:oracle-heavy-vertices} the result follows. Else, since the oracles  in \cref{lemma:basic-oracle} and \cref{lemma:flow-oracle} outputted \textbf{yes} on $X$, since there are no edges longer than $\delta$ in $E(G)\setminus E^*$ and since by construction $V'$ does not contain $(\delta, n)$ heavy vertices,  it must be that
	\begin{align}\label{eq:heavy-vertices-cut-large}
		\Card{E(P_1, P_2, V')\cap \Set{\ij \in E(G)\suchthat \Snorm{v_i-v_j}\leq \delta}}> C^*\cdot \frac{\alpha}{\delta}\cdot \Paren{1+\frac{\ell}{\delta}}\,,
	\end{align}
	for some large enough  constant $C^*>10^{10}C$, where $C>0$ is the universal constant of \cref{lemma:result-heavy-vertices-removal-procedure}.
	
	Let $d=100\cdot C^*\cdot \Paren{\frac{\alpha}{\Card{V(G)}\cdot \delta}\Paren{1+\frac{\ell}{\delta}}}$.
	We compute the maximum $d$-regular flows for each of the partitions $(P_1\cup P_2\,, V')$, $(P_1\,, P_2\cup V')$, $(P_2\,, P_1\cup V')$ as described in \cref{section:background} using the algorithm in \cref{theorem:max-flow-linear-time}. By \cref{eq:heavy-vertices-cut-large} at least one of these cuts has flow $\frac{C^*}{3}\cdot \frac{\alpha}{\delta}\cdot \Paren{1+\frac{\ell}{\delta}}$ as otherwise by duality we have found a $(a/2)$-balanced cut of value at most $C^*\cdot \frac{\alpha}{\delta}\cdot \Paren{1+\frac{\ell}{\delta}}+\Card{E^*}\leq O\Paren{\frac{\alpha}{\delta}\Paren{1+\frac{\ell}{\delta}}}$ as desired.
	Without loss of generality we may always assume this is the partition $(P_1\cup P_2\,, V')$.
	We distinguish two cases:
	\begin{enumerate}
		\item 		$\sum_{\ij \in E(P_1\cup P_2\,, V')\setminus E^*}f_\ij\Snorm{v_i-v_j}\geq \frac{C^*}{10^9}\Paren{\alpha\Paren{1+\frac{\ell}{\delta}}}\,,$
		\item		$\sum_{\ij \in E(P_1\cup P_2\,, V')\setminus E^*}f_\ij\Snorm{v_i-v_j}< \frac{C^*}{10^9}\Paren{\alpha\Paren{1+\frac{\ell}{\delta}}}\,.$
	\end{enumerate}
	Suppose we are in case 1. Let $F$ be the Laplacian of the weighted graph corresponding to the flow and let $D$ be the Laplacian of the complete weighted graph where only edges $\ij$ with $i \in P_1\cup P_2$ and $j \in V'$ have weight $f_\ij$, and the rest have $0$ weight.
	Since we are in case 1 we have
	\begin{align*}
		\iprod{D, X} \geq \frac{C}{10^9}\Paren{\alpha\Paren{1+\frac{\ell}{\delta}}}\,.
	\end{align*}
	Moreover, by definition $\sum_p f_p T_p= F-D\,.$
	Thus we set $x_i=\alpha/\Card{V(G)}$ for all $i \in V(G)$, $f_p$ as in the computed flow for all $p$ and all other variables to $0$. The feedback matrix $Y$ becomes $\frac{\alpha}{\Card{V(G)}}\Id+F-D-F=\frac{\alpha}{\Card{V(G)}}\Id-D$ and we have
	\begin{align*}
		\iprod{\frac{\alpha}{\Card{V(G)}}\Id-D, X}\leq \alpha - \frac{C}{10^9}\cdot \alpha\Paren{1+\frac{\ell}{\delta}}<  -\alpha<0\,.
	\end{align*}
	Moreover notice that $\Norm{\frac{\alpha}{\Card{V(G)}}\Id-D}\leq O\Paren{\frac{\alpha}{\Card{V(G)}}}+d\leq \tilde{O}\Paren{\frac{\alpha}{\Card{V(G)}}}$ where in the last step we used the inequality $\delta\geq \Omega(1/\log n)$.
	In conclusion, in this case the \oracle finds a separating hyperplane and outputs \textbf{no}. 
	Notice also that by construction $D$  has at most $O(m+n)$ non zero entries so the feedback matrix can be computed in time $O(m+n)$.
	
	It remains to consider case 2. 
	By \cref{lemma:result-heavy-vertices-removal-procedure} and Markov's inequality, we know that for any feasible solution $X^*$ to \cref{eq:balanced-cut-sdp} with objective value at most $\alpha$, it holds with probability at least $1/2$:
	\begin{align*}
		\E \Brac{\Card{E(P_1, P_2, V')\cap \Set{\ij \in E(G)\suchthat \Snorm{v_i-v_j}\leq \delta}}}\leq C\cdot \frac{\alpha}{\delta}\cdot \Paren{1+\frac{\ell}{\delta}}\,,
	\end{align*}
	where $C<C^*/ 10^{10}$.
	Thus repeating the procedure $(\log n)^{100}$ times, we get that, for any feasible solution $X^*$ with objective value $\alpha$, with probability at least $1-O(\log n)^{-100}$, for at least one of the resulting partitions $(P_1, P_2\,, V')$
	\begin{align*}
		\Card{E(P_1, P_2\,, V')\cap \Set{\ij \in E(G)\suchthat \Snorm{v_i-v_j}\leq \delta}}\leq C \cdot \frac{\alpha}{\delta}\cdot \Paren{1+\frac{\ell}{\delta}}\,.
	\end{align*} 

	Now consider again our candidate solution $X$ satisfying \cref{eq:heavy-vertices-cut-large}. If after $(\log n)^{100}$ trials we still satisfy \cref{eq:heavy-vertices-cut-large} and are always in case 2, then with probability $1-O(\log n)^{-100}$   there exist at least $\frac{C^*}{10}\cdot \frac{\alpha}{\delta}\cdot \Paren{1+\frac{\alpha}{\delta}}$ edges of length at most $\delta/10^8$ crossing the cut $E(P_1\cup P_2\,, V')\setminus E^*$. 
	
	By design of \cref{algorithm:heavy-vertex-removal-procedure}, then it must be the case that there exists a set of size $\Omega(n)$ of triplets $\Set{i,j,k}\subseteq C$ with $\ij \in E(G)$  and $k\in V^*$  such that $\Snorm{v_i-v_j}\leq \delta/10^8$, $\Snorm{v_j-v_k}\leq \Snorm{v_i-v_k}$ but
	\begin{align*}
		\bbP_{r\overset{u.a.r.}{\sim} [1,2]} \Paren{\Snorm{v_k-v_j}\leq \delta(1+r) \textnormal{ and } \Snorm{v_k-v_i}> \delta(1+r)}\geq 10^4 \cdot \frac{\Snorm{v_i-v_j}}{\delta}\,.
	\end{align*}
	
	Indeed if this scenario does not apply then we would have seen a partition violating \cref{eq:heavy-vertices-cut-large} with probability at least $1-O(\log n)^{-100}$ by the  argument used in the proof of \cref{lemma:heavy-vertex-removal-procedure-probability-cut}.
	
	So suppose this scenario applies, and consider such a triplet $\Set{i,j,k}$. Then we must have
	\begin{align}\label{eq:triplet-violating-triangle-inequality}
		\Snorm{v_i-v_k}\geq \Snorm{v_j-v_k}+ 10^4\cdot \Snorm{v_i-v_j}
	\end{align}
	so we are violating the triangle inequality.
	Furthermore, we know that the sum over each such triplets must satisfy \begin{align*}
		\sum_{\Set{i,j,k} \text{ satisfying \cref{eq:triplet-violating-triangle-inequality}}}\Snorm{v_i-v_j}\geq \Omega\Paren{\frac{\alpha}{\delta}\Paren{1+\frac{\ell}{\delta}}}\,.
	\end{align*} as otherwise with probability $1-O(\log n)^{-100}$ we would have found a cut violating \cref{eq:heavy-vertices-cut-large}. 
	Notice now that we can find such triangle inequalities in linear time by looking at the edges being cut and the vertices being picked at each iteration of \cref{algorithm:heavy-vertex-removal-procedure}.
	
	Thus set $x_i = \alpha/\Card{V(G)}$ for all $i\in V(G)$ and $f_p = \frac{C^* \alpha}{n}$ for $\Theta(n)$ such violated triangle inequalities and a large enough constant $C^*>0$. We set $F=\mathbf{0}$ and 
	\begin{align*}
		\iprod{\frac{\alpha}{\Card{V(G)}}\Id + \sum f_p T_p, X} \leq \alpha -O\Paren{\frac{\alpha}{\delta}} \leq -\alpha < 0\,,
	\end{align*}
	where in the last step we used the assumption $\delta < 1\,.$
	The width of the feedback matrix is at most $O(\alpha/\Card{V(G)})$ and it has $O(m+n)$ entries, thus it can be computed in $O(m+n)$ time.

	Finally we remark that choosing $d=O(\log n)^{100}$ the oracle is $d$-robust by \cref{lemma:jl}.
\end{proof}


\section{The semi-random hierarchical stochastic model}\label{section:semi-random-hsm}
In this section we consider the semi-random hierarchical stochastic model (HSM) from \cite{cohen2019hierarchical} and develop a nearly linear time algorithm that estimates the Dasgupta's cost of the underlying hierarchical clustering model upto constant factor. The main idea is to recursively compute an $O(1)$-approximation to Balanced Cut which produces a graph with $O(1)$-approximation to the Dasgupta's cost~\cite{Dasgupta16}.  Essentially most of this section is directly cited from \cite{cohen2019hierarchical} and we only provide it for the completeness. However, note that using \cref{theorem:main} we can improve the running time of the algorithm to the nearly linear time.  In the following subsection, we formally define the Dasgupta's cost of the graph and the hierarchical stochastic model. 

\subsection{Related notions}
Let $G = (V, E,w)$ be an undirected weighted graph with weight function $w : E \rightarrow \mathbb{R}^{+}$, where $\mathbb{R}^{+}$ denotes non-negative real numbers. For simplicity we let $w(x,y) = w(y,x) = w(\{x,y\})$. For set $U\subseteq V$ we define $G[U]$ to be the subgraph induced by $U$. A hierarchical clustering $T$ of graph $G$ is a rooted binary tree with exactly $|V|$ leaves, such that each leaf is labeled by a unique vertex $x \in V$. 

For $G = (V, E)$ and a hierarchical-clustering tree $T$ we denote the lowest common ancestor of vertex $x$ and $y$ in $T$ by $\text{LCAT} (x,y)$. For any internal node $N$ of $T$, we let $T_N$ to be the subtree of $T$ rooted at $N$ and we define $V(N)$ to be the set of leaves of the subtree rooted at $N$. 
Finally, for a weighted graph $G = (V, E,w)$ and any subset of vertices $A \subseteq V$ we define $w(A) = \sum_{x,y \in A} w(x,y)$, and for any set of edges $E0$, we let $w(E0) =\sum_{e \in E0} w(e)$.For any sets of vertices $A, B \subseteq V$ , we also define $w(A, B) = \sum_{x \in A, y \in B} w(x,y)$. 

Equipped with these notation we define the Dasgupta's cost of a graph for a tree as follows:

\begin{definition}[(Dasgupta's cost\cite{Dasgupta16, cohen2019hierarchical})]
Dasgupta's cost of the tree $T$ for the graph $G = (V, E, w)$
is defined as \[\text{cost}(T; G) = \sum_{(x,y)\in E}\text{leaves}(T[LCA(x; y)])\cdot w(x,y)\text{.}\]
\end{definition}

\begin{definition}[Ultrametric \cite{cohen2019hierarchical}]
A metric space $(X,d)$ is an ultrametric if for every $x,y, z \in X$, $d(x,y) \leq \max\{d(x, z),d(y, z)\}$. 
\end{definition}

We say that a weighted graph $G = (V, E,w)$ is  generated from an ultrametric if there exists an ultrametric $(X,d)$, such that $V \subseteq X$, and for every $x,y \in V$, $x \neq y$, $e = \{x,y\}$ exists, and $w(e) = f (d(x,y))$, where $f : \mathbb{R}^+ \rightarrow \mathbb{R}^+$ is a non-increasing function. For a weighted undirected graph $G = (V, E,w)$ generated from an ultrametric, in general there may be several ultrametrics and
corresponding functions $f$ mapping distances in the ultrametric to weights on the edges, that generate the same graph. It is useful to introduce the notion of a minimal ultrametric that generates $G$. Let $(X,d)$ be an ultrametric that
generates $G = (V, E, w)$ and $f$ the corresponding function mapping distances to similarities. Then we consider the ultrametric $(V,\tilde{d})$ as follows: (i) $\tilde{d}(u,u)=0$ and (ii) for $u\neq v$
\[\tilde{d}(u,v)=\tilde{d}(v,u) =\max_{u', v'}  d(u',v') | f(d(u',v'))=f(d(u,v)) \]

\begin{definition} [Generating Tree \cite{cohen2019hierarchical}] 
\label{ref:gen-tree}
Let $G = (V, E,w)$ be a graph generated by a minimal ultrametric $(V,d)$. Let $T$ be a rooted binary tree with $|V|$ leaves and
$|V | - 1$ internal nodes; let $N$ denote the internal nodes and $L$ the set of leaves of $T$ and let $\sigma : L \rightarrow V$ denote a bijection between the leaves of $T$ and nodes of $V$ . We say that $T$ is a generating tree for $G$, if there exists a weight function $W :\mathcal{N} \rightarrow R^+$, such that for $N_1, N_2 \in \mathcal{N}$, if $N_1$ appears on the path from $N_2$ to the root, $W(N1) \leq W(N2)$. Moreover
for every $x,y \in V$ , $w({x,y}) =W (\text{LCA} T (\sigma^{-1} (x), \sigma^{-1} (y)))$. 
\end{definition}

We say that a graph $G$ is a ground-truth input if it is a graph generated from an ultrametric. Equivalently, there exists a tree $T$ that is generating for $G$.

Now we are ready to define Hierarchical Stochastic Model graphs as follows:
\begin{definition}[Hierarchical Stochastic Model (HSM) \cite{cohen2019hierarchical}].
\label{Definition51}
Let $\widetilde{T}$ be a generating tree for an $n$-vertex
graph $\bar{G}$, called the expected graph, such that all weights are in $[0,1]$. A hierarchical stochastic
model is a random graph $G$ such that for every two vertices $u$ and $v$, the edge $\{u,v\}$ is present independently with probability $w(\{u,v\})=W(\mathrm{LCA}_T (\sigma^{-1}(u), \sigma^{-1}(v)))$, where $w$ and $W$ are the weights functions associated with $\widetilde{T}$ as per Definition \ref{ref:gen-tree}.
\end{definition}

In other words, the probability of an edge being present is given by the weight of the lowest common ancestor of the corresponding vertices in $\widetilde{T}$.

\subsection{The algorithm for the semi-random hierarchical stochastic model}

We generate a random graph, $G = (V, E)$, according to HSM (Definition \ref{Definition51}). The semi-random model  considers a random HSM graph generated as above where an adversary is allowed to only remove edges from $G$. Note that the comparison is to the cost of the generating tree on the graph $\bar{G}$ (Definition \ref{Definition51}). In this section we present the proof of Theorem \ref{theo61}.

Proof of Theorem~\ref{theo61} is a variant of Theorem 6.1 from \cite{cohen2019hierarchical} with nearly-linear running time that uses our fast algorithm for finding Balanced-Cut.

Let $\bar{G}_n=(\bar{V}_n,\bar{E}_n,w)$ be a graph generated according to an ultrametric, where for each $e \in \bar{E}_n$, $w(e) \in (0,1)$ (Definition \ref{ref:gen-tree}). 
Let $G=(V,E)$ be an unweighted random graph with $\left | V \right | = \left | \bar{V}_n \right |=n$ generated from $\bar{G}$ as follows. For every $u,v \in \bar{V}_n$ the edge $(u,v)$ is added to $G$ with probability $w((u,v))$ (Definition \ref{Definition51}). 

We assume that $V = \bar{V}_n$ 
and let $T$ be a generating tree for $\bar{G}$. Let $U \subseteq  V$. Let $\widetilde{T}|_U$ denote the restriction of $T$ to leaves in $U$ 
Let $N(U) $ be the root of $\widetilde{T}|_U$. Consider the following procedure where the nodes appear as leaves in the left and
right subtrees of the root of $\widetilde{T}|_U$. Suppose we follow the convention that the left subtree is never any smaller than the right subtree in $\widetilde{T}|_U$. We say that the canonical node of $\widetilde{T}|_U$ is the first left node $N_L$ encountered in a top-down traversal starting from $N(U)$ such that $(1-b)\cdot |U|\geq V(N_L)\geq b\cdot |U|$, where, $0<b<1/2$ is a constant.
We define $U_L=V(N_L)$, and $U_R=U \setminus  {U_L}$. We say that $(U_L,U_R)$ is the \textit{canonical} cut of $U$. It is easy to see that such a cut always exists since the tree is binary and left subtrees are never smaller than right subtrees. Let $E_{rnd}=\{(u,v)\in E | u \in U_L,v \in U_R\}$.

\begin{lemma}[\cite{cohen2019hierarchical}]
\label{lemma65}
For a random graph $G$ generated as described in Theorem ~\ref{theo61}, with probability at least $1-o(1)$, for every subset $U$ of size at least $n^{2/3}\sqrt{\log n}$, the subgraph $(U, E_{rnd})$ is geometrically expanding up to scale $(1/\sqrt{D}, n, \alpha)$ where
\begin{equation}
\label{eq:eq15}
    \alpha=C.\max\{w(L, R), |U|\cdot D \cdot \log^2 D,|U|\cdot D\cdot \log n\},
\end{equation}
Furthermore, the result also applies in the semi-random setting where an adversary
may remove any subset of edges from the random graph $G$.
\end{lemma}

\begin{theorem}
\label{theorem4}
For any graph $G=(V,E)$, and weight function $w:E \rightarrow \mathbb{R}^{+}$, the $\phi$-\textit{sparsest-cut algorithm} from \cite{cohen2019hierarchical}
outputs a solution of cost at most $O(\phi \cdot OPT)$.
\end{theorem}

Now we show the proof of Theorem~\ref{theo61} which is a variant of Theorem 6.1 from \cite{cohen2019hierarchical} using our nearly-linear time agorithm for Balanced-Cut.

\begin{proof}[Proof of Theorem \ref{theo61}]
Let $b=1/3$. By Theorem \ref{theorem4}, the recursive sparsest cut algorithm approximates Dasgupta's cost up to factor $O(\phi)$ assuming that at every recursion step, the algorithm is provided with a $\phi$-approximation to the $b$-Balanced Cut problem (i.e., minimize cut subject to the constraint that both sides have at least $b$ fraction of vertices being cut)


Note that $cost(\widetilde{T};\bar{G})= \Omega(n^3 \cdot p_{min})=\Omega\left(n^{7/3} \cdot \log n\right)$. Therfore,  once we obtain sets $U$ of size $(n_0=n^{2/3} \cdot \log n)$, since there are at most $n/n_0$ of them, even if we use an arbitrary tree on any such $U$, together this can only add $O(\frac{n}{n_0}
\cdot n^3_0)=O(n^{13/9} \cdot (\log n)^2)=O\left(n^{7/3} \cdot \log n\right)$ to the cost. Thus, we only need to obtain suitable approximations during the recursive procedure as long as $|U| \geq n^{2/3} \cdot \log n$. This
is precisely given by using Lemma~\ref{lemma65}.  Let $D=O(\log n)$, $\delta=\frac{1}{\sqrt{D}}=O\left(\frac{1}{\sqrt{\log n}}\right)$,  let $\kappa\geq \Omega(\sqrt{\log n})$.  Observe that in Equation~\ref{eq:eq15}, $w(L,R)=\Omega(|U|^2 \cdot p_{min})=\Omega \left(n^{2/3}(\log n)^3\right)$, $|U|D \log^2 D = o(|U|D \log n)$, and $D|U|\log n = O\left(n^{2/3}(\log n)^3\right)$. Let $\alpha = O(w(L,R))$. Thus, by Theorem \ref{theorem:balanced-cut-technical} there exists an algorithm that runs in time $\tilde{O}\Paren{\Card{V(G)}+\Card{E(G)}}$ and returns a cut that is an approximation to the $\Omega(b)$-balanced partition $(S,T)$ with cut  of size $\Card{E(S, T)}\leq O(\alpha\cdot  (1+\delta \cdot \kappa\cdot \sqrt{\log n})) = O(\alpha)$ on the induced subgraph of $\bar{G}$ on the vertex set $U$. This observation together with 
the case where subgraphs have size less than $n^{2/3}\log n$ finishes the proof.
\end{proof}

\end{document}